\documentclass[twocolumn, twoside, 10pt]{IEEEtran}

\usepackage{amsmath,amssymb,dsfont,color,bbm}
\usepackage{graphicx,psfrag,cite,subfigure,acronym}
\usepackage{array,enumerate}
\usepackage{float,subfig,rawfonts,mathrsfs,cite}
\usepackage{algorithm,algpseudocode}

\usepackage{wgroup_switch}
\setswitchA{1}
\setswitchB{1}

\newtheorem{proposition}{Proposition}

\newtheorem{definition}{Definition}
\newtheorem{remark}{Remark}
\IEEEoverridecommandlockouts
\ifCLASSINFOpdf \else \fi

\newcommand{\Sagt} {\mathcal{N}_\text{a}}
\newcommand{\Sanc} {\mathcal{N}_\text{b}}
\newcommand{\Nagt} {N_\text{a}}
\newcommand{\Nanc} {N_\text{b}}
\newcommand{\tr} {\text{tr}}
\newcommand{\tx} {\text} 
\newcommand{\T} {^\text{T}}
\newcommand{\M} {\mathbf{M}}
\newcommand{\I} {\mathbf{I}}

\newcommand{\Pkj} {x_{kj}}
\newcommand{\Pk} {x_k}

\newcommand{\Pkjmin} {P_{kj}^{\min}}
\newcommand{\Pkjtot} {P_{kj}^{\max}}
\newcommand{\Pjtot} {P_j^\text{tot}}
\newcommand{\Ptot} {P^\text{tot}}
\newcommand{\Pptot} {P^{'\text{tot}}}
\newcommand{\Je} {\mathbf{J}_\text{e}}
\newcommand{\Jeinv} {\mathbf{J}_\text{e}^{-1}}

\newcommand{\Jr} {\mathbf{J}_\text{r}}

\newcommand{\Q} {\mathbf{Q}_\text{r}}

\newcommand{\rhokj} {\rho_{kj}}
\newcommand{\dest} {\hat{d}}
\newcommand{\phiest} {\hat{\phi}}
\newcommand{\xiest} {\hat{\xi}}
\newcommand{\dwst} {\widetilde{d}}
\newcommand{\phiwst} {\widetilde{\phi}}

\newcommand{\muwst} {{\widetilde{\mu}}}
\newcommand{\xikj} {\xi_{kj}}
\newcommand{\xiwstkj} {\widetilde{\xi}_{kj}}
\newcommand{\uvec} {\mathbf{u}}

\newcommand{\PI} {\mathscr{P}_\tx{1}}
\newcommand{\PISDP} {\PI^\tx{SDP}}
\newcommand{\PR} {\mathscr{P}_\tx{R-0}}
\newcommand{\PRI} {\mathscr{P}_\tx{R-1}}
\newcommand{\PRISDP} {\PRI^\tx{SDP}}
\newcommand{\PRIDI} {\mathscr{P}_{\tx{R-1},k}^\tx{(I)}}
\newcommand{\PRIDII} {\PRI^\tx{(II)}}

\newcommand{\PII} {\mathscr{P}_\tx{2}}
\newcommand{\PIICP} {\PII^\tx{SOCP}}
\newcommand{\PRII} {\mathscr{P}_\tx{R-2}}
\newcommand{\PRIICP} {\PRII^\tx{SOCP}}
\newcommand{\PRIIDI} {\PRII^\tx{(I)}}
\newcommand{\PRIIDII} {\PRII^\tx{(II)}}

\begin{document}

\title{Robust Power Allocation for Energy-Efficient Location-Aware Networks}

\author{\IEEEauthorblockN{William~Wei-Liang~Li,~\IEEEmembership{Member,~IEEE,} Yuan~Shen,~\IEEEmembership{Student~Member,~IEEE,} Ying~Jun~(Angela)~Zhang,~\IEEEmembership{Senior~Member,~IEEE,} and Moe~Z.~Win,~\IEEEmembership{Fellow,~IEEE}}
\\[-0.5em]
    \thanks{Manuscript submitted December 20, 2011; revised June 19, 2012, and December 21, 2012; accepted December 21, 2012. This research was supported, in part, by the GRF grant (Project number 419509) established under the University Grants Committee (UGC) of Hong Kong Special Administrative Region, the National Science Foundation under Grant ECCS-0901034, the Office of Naval Research under Grant N00014-11-1-0397, and MIT Institute for Soldier Nanotechnologies. This paper was presented in part at the IEEE International Conference on Communications, Kyoto, Japan, June 2011 and Ottawa, Canada, June 2012.}
	\thanks{W.~W.-L.~Li was with Department of Information Engineering, The Chinese University of Hong Kong, Shatin, New Territories, Hong Kong, and is now with Department of Electrical and Computer Engineering, University of California, Santa Barbara, CA 93106 USA (e-mail: {wlli@ieee.org}).}
	\thanks{Y.~J.~(A)~Zhang is with Department of Information Engineering, The Chinese University of Hong Kong, Shatin, New Territories, Hong Kong (e-mail: {yjzhang@ie.cuhk.edu.hk}).}
	\thanks{Y.~Shen and M.~Z.~Win are with the Laboratory for Information and Decision Systems (LIDS), Massachusetts Institute of Technology, 77 Massachusetts Avenue, Cambridge, MA 02139 USA (e-mail: {shenyuan@mit.edu}, {moewin@mit.edu}).}
}

\maketitle 

\markboth{IEEE/ACM Transactions on Networking Month Date year}{Li \MakeLowercase{\textit{et al.}}: Robust Power Allocation for Energy-Efficient Location-Aware Networks}



\setcounter{page}{1}

\vspace{-0.5cm}
\begin{abstract}
In wireless location-aware networks, mobile nodes (agents) typically obtain their positions through ranging with respect to nodes with known positions (anchors). Transmit power allocation not only affects network lifetime, throughput, and interference, but also determines localization accuracy. In this paper, we present an optimization framework for robust power allocation in network localization to tackle imperfect knowledge of network parameters. In particular, we formulate power allocation problems to minimize the squared position error bound (SPEB) and the maximum directional position error bound (mDPEB), respectively, for a given power budget. We show that such formulations can be efficiently solved via conic programming. Moreover, we design an efficient power allocation scheme that allows distributed computations among agents. The simulation results show that the proposed schemes significantly outperform uniform power allocation, and the robust schemes outperform their non-robust counterparts when the network parameters are subject to uncertainty.
\end{abstract}


%
%
%
%

%
%

\section{Introduction}\label{sec:intro}


Positional information is of critical importance for future wireless networks, which will support an increasing number of location-based applications and services \cite{WinConMazSheGifDarChi:J11,SayTarKha:05,PahLiMak:02,PatAshKypHerMosCor:05,SheWin:J10a,SheWymWin:J10, GezTiaGiaKobMolPooSah:05,
VerDarMazCon:B08,WymLieWin:J09}. Example applications include cellular positioning, search and rescue work, blue-force tracking, etc., covering civilian life to military operations. In GPS-challenged environments, wireless network localization typically refers to a process that determines the positions of mobile nodes (agents) based on the measurements with respect to mobile/static nodes with known positions (anchors), as illustrated in Fig.~\ref{fig:topo}. With the rapid development of advanced wireless techniques, wireless network localization has attracted numerous research interests in the past decades \cite{DarConBurVer:07,DarConFerGioWin:J09,PaoGioChiMinMon:08,RabOppDen:06a,ConGueDarDecWin:J12,YuMonRabCheOpp:06,MazLorBah:J10,SheMazWin:J12,NicFon:09,KhaKarMou:09,KhaKarMou:10}.


\begin{figure}[t!]
	\vspace{1em}
	\centering
	\includegraphics[height=4cm]{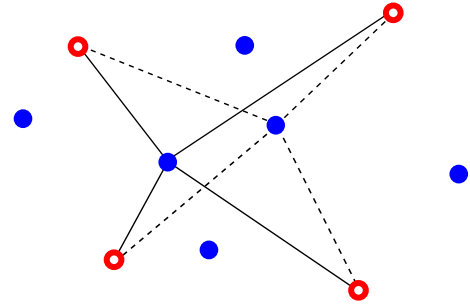}
	\caption{\label{fig:topo}Location-aware networks: the anchors (red circle) localize the agents (blue dot) based on inter-node range measurements.}
	\vspace{-1em}
\end{figure}


Localization accuracy is a critical performance measure of wireless location-aware networks. In recent work \cite{SheWin:J10a,SheWymWin:J10}, the fundamental limits of wideband localization have been derived in terms of the squared position error bound (SPEB) and directional position error bound (DPEB). It shows that localization accuracy is related to several aspects of design, including network topology, signal waveforms, and transmit power. 
Power allocation for wireless network localization plays a critical role in reducing localization errors or energy consumption, when the nodes are subject to limited power resources or quality-of-service (QoS) requirements \cite{MesPooSch:07,GorKinKymRubWanZus:10,GorWalZus:11}. 
Optimal or near-optimal trade-off between localization errors and energy consumption can be obtained by optimization methods, which have played an important role in maximizing communication and networking performance under limited resources \cite{AhuMagOrl:93, KelMauTan:98, BoyVan:04, LuoYu:06, LiZhaSoWin:J10, HuaRao:13, Fos:01,She:01}. The authors in \cite{SheWin:C08a} formulated several optimization problems for anchor power allocation in wideband localization systems, and derived the optimal solution for single-agent networks. In \cite{LiSheZhaWin:C11}, it exploited the geometrical interpretation of localization information to minimize the maximum DPEB (mDPEB).\footnote{The mDPEB characterizes the maximum position error of an agent over all directions.} In \cite{GodPetPoo:11}, it investigated the localization using MIMO radar systems, and adopted the constraint relaxation and domain decomposition methods to obtain sub-optimal solutions for power allocation. In general, how to optimally allocate the transmit power in location-aware networks still remains as an open problem.


Power allocation schemes should be adapted to the instantaneous network conditions, such as network topology and channel qualities, for optimizing the localization performance. Previous work 
on power allocation in location-aware networks
assumes that the network parameters such as nodes' positions and channel conditions are perfectly known \cite{SheWin:C08a, LiSheZhaWin:C11, GodPetPoo:11}. However, these parameters are obtained through estimation and hence subject to uncertainty. The power allocation based on imperfect knowledge of network parameters often leads to sub-optimal or even infeasible solutions in realistic networks \cite{BerBroCar:11, BenGhaNem:B09, QueWinChi:J10}. Therefore, it is essential to design a robust scheme to combat the uncertainty in network parameters.


In this paper, we present an optimization framework for robust power allocation in network localization to tackle imperfect knowledge of network parameters. Specifically, we treat the fundamental limits of localization accuracy, i.e., SPEB and mDPEB, as the performance metrics. The main contributions are summarized as follows.
\begin{itemize}
	\item We formulate optimization problems for power allocation to minimize SPEB/mDPEB subject to limited power resources, and prove that these formulations can be transformed into conic programs.\footnote{Conic programs can be efficiently solved by off-the-shelf optimization tools \cite{LuoStuZha:00,LuoYu:06}}
	\item We propose a robust optimization method for the worst-case SPEB/mDPEB minimization in the presence of parameter uncertainty. The proposed robust formulations retain the same form of conic programs as their non-robust counterparts.
	\item We develop a distributed algorithm for robust power allocation, which decomposes the original problem into several subproblems enabling parallel computations among all the agents without loss of optimality.
\end{itemize}


The rest of the paper is organized as follows. In Section \ref{sec:model}, we describe the system model and introduce the performance metrics. In Section \ref{sec:optimal}, we formulate the power allocation problems into conic programs. In Section \ref{sec:robust}, robust power allocation schemes are proposed to combat the uncertainty in network parameters. In Section \ref{sec:distributed}, we further decompose our robust formulation into several subproblems that can be independently solved by each agent. In Section \ref{sec:simu}, the performance of the proposed schemes is investigated through simulations. Finally, the paper is concluded in Section \ref{sec:conclude}.

\emph{Notations:} We use lowercase and uppercase bold symbols to denote vectors and matrices, respectively; $\det(\mathbf{A})$ and $\text{tr}(\mathbf{A})$ denote the determinant and trace of matrix $\mathbf{A}$, respectively; the superscript $(\cdot)\T$ and $\|\cdot\|$ denote the transpose and Euclidean norm of its argument, respectively; matrices $\mathbf{A}\succeq\mathbf{B}$ denotes that $\mathbf{A}-\mathbf{B}$ is positive semidefinite. We define the unit vector $\uvec(\phi)=[\,\cos\phi \,~\sin\phi\,]\T$. We use calligraphic symbols, e.g., $\mathcal{N}$, to denote sets, and $\mathbb{E}\{\cdot\}$ and $\Pr\{\cdot\}$ to denote the expectation and probability operators, respectively.

\section{System Model}\label{sec:model}
In this section, we describe the system model, and introduce two performance metrics of location-aware networks.

\subsection{Network Settings}
Consider a 2-D location-aware network consisting of $\Nagt$ agents and $\Nanc$ anchors, where the sets of agents and anchor are denoted by $\Sagt=\{1,2,\ldots,\Nagt\}$ and  $\Sanc=\{\Nagt+1,\Nagt+2,\ldots, \Nagt+\Nanc\}$, respectively. The 2-D position of node $k$ is denoted by $\mathbf{p}_k$. The angle and distance between nodes $k$ and $j$ are given by $\phi_{kj}$ and $d_{kj}$, respectively. The anchors are mobile/static nodes with known positions, and subject to limited power resources. The agents aim to determine their positions based on the radio signals transmitted from the anchors. For instance, agents can obtain the signal metrics such as time-of-arrival (TOA) from the received signals, and then calculate their positions via triangulation \cite{SheWin:J10a}.

The multipath received waveform at agent $k$ from anchor $j$ is modeled as \cite{SheWin:J10a}
\begin{equation}\label{eq:waveform}
	r_{kj}(t) = \sum_{l=1}^{L_{kj}} \sqrt{\Pkj} \cdot  \alpha_{kj}^{(l)} \, s\big(t-\tau_{kj}^{(l)}\big) + z_{kj}(t), \; t\in[0,T_\text{ob})\!
\end{equation}
where $\Pkj$ is the power of the transmit waveform from anchor $j$ to agent $k$, $s(t)$ is a known transmit waveform, $\alpha_{kj}^{(l)}$ and $\tau_{kj}^{(l)}$ are the amplitude and delay, respectively, of the $l$th path, $L_{kj}$ is the number of multipath components, $z_{kj}(t)$ represents additive white Gaussian noise (AWGN) with two-side power spectral density $N_0/2$, and $[0,T_\text{ob})$ is the observation interval. 

We consider that the measurements between anchors and agents do not interfere each other by using medium access control, and the network is synchronized such that the inter-node distance is estimated using one-way time-of-flight (TOF).\footnote{There are two common ways for inter-node distance estimation based on TOA: one-way TOF (only anchor transmits) or round-trip TOF (both anchor and agent transmit). The former requires anchors and agents to be synchronized for  distance estimation.} Our work can be extended to asynchronous networks where round-trip TOF is employed for distance estimation, and it will be discussed in Section \ref{sec:optimal}.

\subsection{Position Error Bound}
The SPEB introduced in \cite{SheWin:J10a} is a performance metric that characterizes the localization accuracy, defined as
\begin{equation}\label{eq:SPEB}
	\mathcal{P}(\mathbf{p}_k) \triangleq \tr\big\{\Jeinv(\mathbf{p}_k;\{\Pkj\})\big\}
\end{equation}
where $\Je(\mathbf{p}_k;\{\Pkj\})$ is the equivalent Fisher information matrix (EFIM) for agent $k$'s position $\mathbf{p}_k$. Using the information inequality \cite{Van:B68}, we can show that the squared position error is bounded below as
\begin{equation*}
	\mathbb{E}\left\{\|\hat{\mathbf{p}}_k-\mathbf{p}_k\|^2\right\}
	\geq \mathcal{P}(\mathbf{p}_k)
\end{equation*}
where $\hat{\mathbf{p}}_k$ is an unbiased estimate of the position $\mathbf{p}_k$. The EFIM in \eqref{eq:SPEB} can be derived based on the received waveform in \eqref{eq:waveform} as a $2\times 2$ matrix \cite{SheWin:J10a}
\begin{equation}\label{eq:EFIM}
	\Je(\mathbf{p}_k;\{\Pkj\}) = \sum_{j\in\Sanc} \xikj \, \Pkj \, \Jr(\phi_{kj})
\end{equation}
where $\Jr(\phi_{kj}) = \uvec(\phi_{kj}) \uvec(\phi_{kj})\T$ is a $2\times 2$ matrix, 
and $\xikj$ is a positive coefficient determined by the channel properties, given by,\footnote{The derivation of $\xikj$ is given in \cite{SheWin:J10a}, and this parameter can be obtained through channel estimation.}
\begin{equation}\label{eq:ch-para}
	\xikj = \frac{8\pi^2 W^2}{c^2} (1-\chi_{kj}) \frac{(\alpha_{kj}^{(1)})^2}{N_0}
\end{equation}
with $W$ as the effective bandwidth, $c$ as the light speed, $\chi_{kj}$ as path-overlap coefficient characterizing the effect of multipath propagation for localization, $N_0$ as the noise spectrum density.\footnote{Although the structure of SPEB is derived based on the received waveforms for wideband systems in \cite{SheWin:J10a}, it is also observed in other TOA- or RSS-based localization systems, e.g., \cite{JouDarWin:J08,GodHaiBlu:10, QiKobSud:06, MazLorBah:J10}.}

Since the SPEB characterizes the fundamental limit of localization accuracy and is achievable in high SNR regimes, we will use it as a performance metric for location-aware networks, and allocate the transmit power to optimize the system performance by minimizing the SPEB. 


%

\subsection{Directional Decoupling of SPEB}\label{sec:mDPEB}
We then introduce the notations of DPEB and mDPEB \cite{SheWymWin:J10}.
The EFIM \eqref{eq:EFIM} can be written, by eigen decomposition, as
\begin{align*}
	\Je(\mathbf{p}_k;\{\Pkj\}) 
	&= \mathbf{U}_{\theta_{k}} 
	\begin{bmatrix} \mu_{1,k} & 0\\ 0 & \mu_{2,k} \end{bmatrix}
	\mathbf{U}_{\theta_{k}}\T
\end{align*}
where $\mu_{1,k}$ and $\mu_{2,k}$ are the ordered eigenvalues of EFIM ($\mu_{1,k}\geq\mu_{2,k}$), given by
\begin{align*}
	\mu_{1,k}, \mu_{2,k}
	=\frac{1}{2}\bigg( \sum_{j\in\Sanc} \xikj \, \Pkj
	\pm \Big\| \sum_{j\in\Sanc} \xikj \, \Pkj \,\uvec(2\phi_{kj})
	\Big\| \bigg)
\end{align*}
and $\mathbf{U}_{\theta_{k}}$ is a rotation matrix with angle $\theta_{k}$, given by
\begin{equation*}
	\mathbf{U}_{\theta_{k}} = \begin{bmatrix}\cos\theta_{k} & -\sin\theta_{k} \\
	\sin\theta_{k} & \cos\theta_{k} \end{bmatrix}.
\end{equation*}

Geometrically, the EFIM for agent $k$ can be viewed as an information ellipse given by $\{\mathbf{z}\in\mathbb{R}^2: \mathbf{z}\T \Je^{-1}(\mathbf{p}_k;\{\Pkj\})  \mathbf{z} =1\}$ (see Fig. \ref{fig:ellipse}), where $2\sqrt{\mu_{1,k}}$ and $2\sqrt{\mu_{2,k}}$ give the major axis and minor axis, respectively. 

\begin{definition}
	The directional position error bound (DPEB) of agent $k$ along the direction $\varphi$ is defined as
	\begin{equation*}
		\mathcal{P}(\mathbf{p}_k;\varphi) \triangleq 
		\uvec(\varphi)\T [\Jeinv(\mathbf{p}_k,\{\Pkj\})] \uvec(\varphi).
	\end{equation*}
\end{definition}

\begin{proposition}\label{thm:mDPEB}
	The mDPEB of agent $k$ is 
	\begin{equation}\label{eq:mDPEB}
		\max_{\varphi\in[0,2\pi)} \left\{\mathcal{P}(\mathbf{p}_k;\varphi)\right\} = \frac{1}{\mu_{2,k}}.
	\end{equation}
\end{proposition}
\begin{proof}
	See Appendix \ref{apd:mDPEB}.
\end{proof}

Proposition \ref{thm:mDPEB} can also be understood via the information ellipse  of EFIM. The information for localization achieves the maximum along the major axis and the minimum along the minor axis. Due to the reciprocal, the SPEB is dominated by the mDPEB, which is the inverse of the smaller eigenvalue of the EFIM. Therefore, in order to improve the localization performance, it is more helpful to maximize the smaller eigenvalue of EFIM, equivalently to minimize the mDPEB that characterizes the maximum position error of an agent over all directions. We will use mDPEB as another performance metric of localization accuracy.


\begin{figure}[t!]
	\centering
	\psfrag{x}[][]{x}
	\psfrag{y}[][]{y}
	\psfrag{mu1}[][]{\hspace{0.1cm}$\sqrt{\mu_{1,k}}$}
	\psfrag{mu2}[][]{\hspace{-0.2cm}$\sqrt{\mu_{2,k}}$}
	\psfrag{beta}[][]{$\theta_k$}
	\vspace{0.5em}
	\includegraphics[height=4cm]{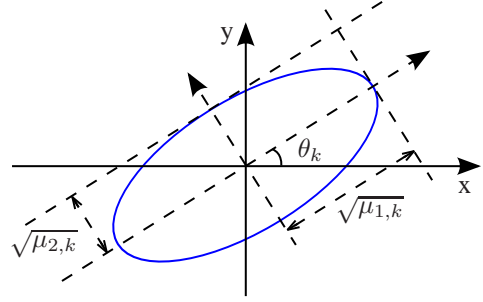}
	\caption{\label{fig:ellipse}Geometrical interpretation of the EFIM for agent $k$.}
\end{figure}

\section{Optimal Power Allocation via \\Conic Programming}\label{sec:optimal}
In this section, we formulate the power allocation problem using SPEB and mDPEB as the objective functions, respectively. We show that the SPEB minimization is a semidefinite program (SDP) and the mDPEB minimization is a second-order conic program (SOCP). 

\subsection{Problem Formulation Based on SPEB}
We first consider the problem of optimal power allocation that minimizes the total SPEB while the network is subject to a budget of power consumption. The problem can be formulated as\footnote{The structure of the problem retains with additional linear constraints, such as the maximum transmit power from anchor $j$ to agent $k$, and the maximum total transmit power from anchor $j$. See Remark \ref{rmk:linear} for more discussion.}
\begin{align}
\PI: \quad	
\min_{\{\Pkj\}} \quad& \sum_{k\in\Sagt} \tr\big\{\Jeinv(\mathbf{p}_k;\{\Pkj\})\big\}
	\label{eq:P-obj}
	\\
	\text{s.t.} \quad
			& \sum_{k\in\Sagt} \sum_{j\in\Sanc} \Pkj \leq \Ptot
			\label{eq:P-con-ttl}
			\\
			& \Pkj \geq 0, \quad \forall k\in\Sagt, ~\forall j\in\Sanc
			\label{eq:P-con-link}
\end{align}
where 
\eqref{eq:P-con-ttl} gives the total transmit power budget $\Ptot$ for all the anchors. We first show the convexity of the above problem in the following proposition.
\begin{proposition}\label{thm:convexity}
	The problem $\PI$ is convex in $\Pkj$.
\end{proposition}

\begin{proof}
	See Appendix \ref{apd:convexity}.
\end{proof}

Since $\PI$ is a convex problem, the optimal solution can be achieved by the standard convex optimization algorithms, e.g., interior point method. We next show that such problem can be converted to a SDP problem, which is a more favorable formulation since many fast real-time optimization solvers are available for SDP \cite{VanBoy:96,LuoMaSoYeZha:10}.

To obtain an equivalent formulation to $\PI$, we replace the EFIMs in \eqref{eq:P-obj} with auxiliary matrices $\M_k$, and add another constraint
$$
\M_k \succeq \Jeinv(\mathbf{p}_k;\{\Pkj\}).
$$
Since $\Je(\mathbf{p}_k)$ is a positive semidefinite matrix, due to the property of Schur complement, the above inequality is equivalent to
$$
\begin{bmatrix}
	\M_k & \I \\ 
	\I & \Je(\mathbf{p}_k;\{\Pkj\})
\end{bmatrix} \succeq 0 \,.
$$
Then, we can obtain a SDP formulation $\PISDP$ equivalent to $\PI$,
\begin{align}
\PISDP\!\!: 
\min_{\{\Pkj\},\,\M_k} \,& \sum_{k\in\Sagt} \tr\left\{\M_k\right\} 
	\notag\\
	\text{s.t.} \quad\,& \!\begin{bmatrix}
							\M_k & \I \\ 
							\I & \Je(\mathbf{p}_k;\{\Pkj\})
					   \end{bmatrix} \succeq 0\,, \,\forall k\in\Sagt
				\notag\\
				& \eqref{eq:P-con-ttl}~\tx{--}~\eqref{eq:P-con-link}.
				\notag
\end{align}
Hence, the optimal solution of $\PI$ can be efficiently obtained by solving the SDP formulation $\PISDP$.

\subsection{Problem Formulation Based on mDPEB}


We now consider the minimization of total mDPEB as our objective. The problem can be formulated as
\begin{align}
	\PII: \quad	
	\min_{\{\Pkj\}} \quad& \sum_{k\in\Sagt} \frac{1}{\mu_{2,k}}
	\notag\\
	\text{s.t.}	\quad
				& \eqref{eq:P-con-ttl}~\tx{--}~\eqref{eq:P-con-link}
				\notag
\end{align}
which can be equivalently converted to
\begin{align}
	\PIICP\!\!: 
	\min_{\{\Pkj,r_k\}} \,& \sum_{k\in\Sagt} 
	\frac{1}{\sum_{j\in\Sanc} \xikj \, \Pkj - r_k}
	\notag\\
	\text{s.t.} \quad& r_k \!\geq\! 
				\Big\|  \sum_{j\in\Sanc} \xikj \, \Pkj \, \uvec(2\phi_{kj})
				\Big\|, \;\forall k\in\Sagt
				\label{eq:PmD-CP-con-soc}\\
				& \eqref{eq:P-con-ttl}~\tx{--}~\eqref{eq:P-con-link}.
				\notag
\end{align}
The constraints \eqref{eq:PmD-CP-con-soc} define $\Nagt$ second-order cones given by
\begin{align*}
	\mathcal{Q}_k=\{(r_k,\mathbf{z}_k)\in\mathbb{R}\times\mathbb{R}^2:
	r_k\geq\|\mathbf{z}_k\|\}, \quad \forall k\in\Sagt
\end{align*}
where $\mathbf{z}_k= \sum_{j\in\Sanc} \xikj \, \Pkj \, \uvec(2\phi_{kj})$.
Moreover, the objective is convex in $\{\Pkj,r_k\}$, since the denominator is a positive linear combination of $\{\Pkj,r_k\}$, and the reciprocal is a convex and decreasing function which preserves convexity \cite{BoyVan:B04}.
Thus, we obtain a nonlinear SOCP problem which is convex in $\Pkj$.
\begin{remark}
	We consider a general model where each anchor can use different transmit power, and our work can be applied to the anchor broadcasting scenario by simply adding constraint $\Pkj = x_j,~\forall k\in\Sagt$.
\end{remark}


\begin{remark}\label{rmk:linear}
Additional linear constraints on transmit power can be imposed depending on the realistic requirements of location-aware networks. For example, we can consider $\Pkjmin \leq \Pkj \leq \Pkjtot$ where $\Pkjmin$ and $\Pkjtot$ are the lower and upper limit of the transmit power from anchor $j$ to agent $k$, respectively; or $\sum_{k\in\Sagt} \Pkj \leq \Pjtot$ where $\Pjtot$ is the upper limit of the total transmit power from anchor $j$. Due to the linearity of these constraints, the convexity of the problem is retained, and the optimal solution can be obtained via conic programming.
\end{remark}


\begin{remark}\label{rmk:async-opt}
	For the asynchronous networks where round-trip TOF is employed for distance estimation, we need to allocate the transmit power for both anchors and agents. Let $\Pkj'$ denote the power of the transmit waveform from agent $k$ to anchor $j$. In addition to the total anchor power constraint in \eqref{eq:P-con-ttl}, we also impose a total power constraint on agents, i.e.,
\begin{align}\label{eq:Pp-con-ttl}
	\sum_{k\in\Sagt} \sum_{j\in\Sanc} \Pkj' \leq \Pptot
\end{align}
where
\begin{align}\label{eq:Pp-con-link}
	\Pkj' \geq 0, \quad \forall k\in\Sagt, ~\forall j\in\Sanc.
\end{align}
It can be shown that the EFIM of agent $k$ is given by
\begin{align*}
	\Je(\mathbf{p}_k;\{\Pkj\}) = \sum_{j\in\Sanc} \xikj \, g(\Pkj,\Pkj') \, \Jr(\phi_{kj})
\end{align*}
where the equivalent power $g(\Pkj,\Pkj') = 4 \big(\Pkj^{-1}+\Pkj'^{\,-1}\big)^{-1}$. 
To derive the maximum total equivalent power, we consider the following problem
\begin{align*}
	\max_{\{\Pkj,\Pkj'\}} ~& \sum_{k\in\Sagt} \sum_{j\in\Sanc} g(\Pkj,\Pkj') \\
	\text{s.t.} \quad
	& \eqref{eq:P-con-ttl}~\tx{--}~\eqref{eq:P-con-link} \\
	& \eqref{eq:Pp-con-ttl}~\tx{--}~\eqref{eq:Pp-con-link}.
\end{align*}
Using the Karush-Kuhn-Tucker conditions \cite{BerNedOzd:B03}, it can be proved that the optimal value is reached as a constant $g(\Ptot,\Pptot)$ if and only if 
\begin{align}
	\Pkj' = \frac{\Pptot}{\Ptot} \, \Pkj.
\end{align}

Hence, in order to achieve the maximum total equivalent power, the power allocated on anchors and agents should be proportional and consequently, the EFIM for asynchronous network is
\begin{align*}
	\Je(\mathbf{p}_k;\{\Pkj\}) = \sum_{j\in\Sanc} \xikj \, \frac{4 \Pptot}{\Pptot+\Ptot} \, \Pkj \, \Jr(\phi_{kj})
\end{align*}
which is with the same structure as the EFIM of synchronous network in \eqref{eq:EFIM}. Therefore, the power allocation on both anchors and agents in asynchronous networks can be equivalently converted into anchor power allocation in synchronous networks.
\end{remark}

\subsection{Formulations with QoS Guarantee}\label{sec:optimal-QoS}

We next briefly show that the proposed framework also applies to another two types of problem formulations based on different QoS requirements.

\subsubsection{Energy-efficient Formulation}
The objective is to minimize the total transmit power subject to the requirements for agents' SPEBs, i.e.,
\begin{align}
	\min_{\{\Pkj\}} \quad& \sum_{k\in\Sagt} \sum_{j\in\Sanc} \Pkj
	\notag\\
	\text{s.t.} \quad& \tr\left\{\Jeinv(\mathbf{p}_k;\{\Pkj\})\right\}
	\leq \gamma_k, \quad \forall k\in\Sagt \label{eq:con-speb}
	\\
	& \eqref{eq:P-con-link}.
	\notag
\end{align}
Similarly, a formulation for the mDPEB case can be obtained by replacing \eqref{eq:con-speb} with
\begin{align}\label{eq:con-mdpeb}
	\frac{1}{\mu_{2,k}}
	\leq \gamma_k\,, \quad \forall k\in\Sagt.
\end{align}

\subsubsection{Min-max SPEB Formulation}
The objective is to minimize the maximum SPEB among all the agents, i.e.,
\begin{align*}
	\min_{\{\Pkj\}} \quad& \max_{k}\Big\{\tr\left\{\Jeinv(\mathbf{p}_k;\{\Pkj\})\right\}\!\Big\}
	\\
	\text{s.t.} \quad& \eqref{eq:P-con-ttl}~\tx{--}~\eqref{eq:P-con-link}.
\end{align*}
It can be equivalently transformed into
\begin{align*}
	\min_{\{\Pkj\},\,\gamma} \quad& \gamma
	\notag\\
	\text{s.t.} \quad& \tr\left\{\Jeinv(\mathbf{p}_k;\{\Pkj\})\right\}
	\leq \gamma\,, \quad \forall k\in\Sagt
	\\
	& \eqref{eq:P-con-ttl}~\tx{--}~\eqref{eq:P-con-link}
\end{align*}
which turns out to be with the same structure as the energy-efficient formulation. Similarly, a min-max formulation for the mDPEB case can be obtained by replacing the SPEB with the mDPEB in the constraint.

Note that since the above formulations with QoS guarantee have the same structure as $\PI$ or $\PII$, which can be solved efficiently by conic programing, we will focus on $\PI$ and $\PII$ in the following.

To obtain the optimal solutions of $\PI$ and $\PII$, it requires the network parameters, i.e., the channel parameter $\xikj$ and the angle $\phi_{kj}$. However, $\xikj$'s and $\phi_{kj}$'s are usually not perfectly known in realistic networks, and only estimated values are available. 
When $\xikj$'s and $\phi_{kj}$'s are subject to uncertainty, the formulation $\PI$ or $\PII$ may fail to provide reliable solutions, since the actual SPEB/mDEPB is not necessarily minimized. Therefore, it is essential to design a power allocation scheme which is robust to the uncertainty in network parameters.

\section{Robust Power Allocation under \\Imperfect Knowledge of Network Parameters}\label{sec:robust}

In this section, we consider the location-aware networks with imperfect knowledge of network parameters, and propose robust optimization methods to minimize the worst-case SPEB/mDPEB under parameter uncertainty.

\subsection{Robust Counterpart of SPEB Minimization}

In realistic location-aware networks, the network parameters, i.e., $\xikj$ and $\phi_{kj}$, can be obtained through channel estimation or inferred based on the prior information of agents' positions,\footnote{The prior position information is available in applications such as navigation.} and hence are both subject to uncertainty. We adopt robust optimization methodology, which is developed in recent years to handle the optimization problems with data uncertainty \cite{BenGhaNem:B09}. Typically, the data defining the optimization problem is assumed to lie in a certain bounded set, referred to as \emph{uncertainty set}.
Here we consider the actual channel parameters and angles lie in linear uncertainty sets, i.e.,\footnote{We consider the parameter $\xikj$ related to the channel properties to be always positive, i.e., $\xiest_{kj} - \varepsilon^\xi_{kj} > 0$.}
\begin{align*}
	& \xikj \in \mathcal{S}^\xi_{kj} \triangleq [\, \xiest_{kj} - \varepsilon^\xi_{kj}\,, \,\xiest_{kj} + \varepsilon^\xi_{kj} \,] \\ 
	& \phi_{kj} \in \mathcal{S}^\phi_{kj} \triangleq [\,  \phiest_{kj} - \varepsilon^{\phi}_{kj}\,,\, \phiest_{kj} + \varepsilon^{\phi}_{kj} \,]
\end{align*}
where $\xiest_{kj}$ and $\phiest_{kj}$ denote channel parameter and angle with uncertainty, respectively, and $\varepsilon^\xi_{kj}$ and $\varepsilon^{\phi}_{kj}$ are both small positive numbers denoting the maximum uncertainty in the channel parameter and angle, respectively.\footnote{If uncertainty exists in anchor positions, it can be equivalently converted into the uncertainty in channel qualities \cite{SheWymWin:J10}.}

To deal with the network parameter uncertainty, we adopt robust optimization techniques to guarantee the worst-case performance. Instead of using the estimated values, we consider minimizing the largest SPEB over the possible set of actual network parameters, i.e.,
\begin{align*}
	\PR: \quad	
	\min_{\{\Pkj\}} \,& \max_{\{\xikj\in\mathcal{S}^\xi_{kj},\,\phi_{kj}\in\mathcal{S}^{\phi}_{kj}\}} \, \sum_{k\in\Sagt}\tr\left\{\Jeinv(\mathbf{p}_k;\{\Pkj\})\right\}
	\\
	\text{s.t.} \quad
				& \eqref{eq:P-con-ttl}~\tx{--}~\eqref{eq:P-con-link}.
\end{align*}

Since $\tr\left\{\Jeinv(\mathbf{p}_k;\{\Pkj\})\right\}$ is a monotonically non-increasing function of $\xikj$, the maximum SPEB over $\xikj$ is independent of $\phi_{kj}$. Hence, the maximization over $\xikj$ simply follows that
\begin{align*}
	\xiwstkj\triangleq 
	\arg\max_{\{\xikj\in\mathcal{S}^\xi_{kj}\}} ~ \tr\left\{\Jeinv(\mathbf{p}_k;\{\Pkj\})\right\} = 
	\xiest_{kj} - \varepsilon^\xi_{kj}.
\end{align*}
On the other hand, however, the maximization over $\phi_{kj}$ is not trivial, because 
\begin{align}\label{eq:maxphi}
	\{\phiwst_{kj}\} & \triangleq  
	 \arg\max_{\{\phi_{kj}\in\mathcal{S}^{\phi}_{kj}\}} ~ \tr\left\{\Jeinv(\mathbf{p}_k;\{\Pkj\})\right\} \notag\\
	& = 
	\arg\max_{\{\phi_{kj}\in\mathcal{S}^{\phi}_{kj}\}} ~ \Big\|
	\sum_{j\in\Sanc} \xikj \, \Pkj \, \uvec(2\phi_{kj})
	\Big\|^2 
\end{align}
and the right-hand side of \eqref{eq:maxphi} is not a convex problem. Hence, it is difficult to obtain a close-form solution of $\{\phiwst_{kj}\}$ since it depends on $\{\Pkj\}$. 

We next consider a relaxation for the robust optimization with respect to $\{\phi_{kj}\}$ and introduce a new matrix
\begin{equation}\label{eq:matQ}
	\Q(\phiest_{kj},\delta_{kj}) = \Jr(\phiest_{kj}) - \delta_{kj}\cdot\I
\end{equation}
to replace $\Jr(\phi_{kj})$ in the SPEB in \eqref{eq:SPEB}. We will show that the worst-case SPEB over $\phi_{kj}$ can be bounded above by the new function for sufficiently large $\delta_{kj}$. The details are given in the following proposition.

\begin{proposition}\label{thm:Qmatrix}
	If $\sum_{j\in\Sanc} \xikj \, \Pkj \, \Q(\phiest_{kj},\delta_{kj}) \succeq 0$ and $\delta_{kj} \geq \sin\varepsilon^{\phi}_{kj}$, the maximum SPEB over the actual angle $\phi_{kj}$ is always upper bounded as
	\begin{align}\label{eq:wstbnd}
		&\max_{\{\phi_{kj}\in\mathcal{S}^{\phi}_{kj}\}}
		\tr\big\{\Jeinv(\mathbf{p}_k;\{\Pkj\})\big\} \notag\\
		&\hspace{4em} \leq ~
		\text{\tr}\bigg\{\Big(\sum_{j\in\Sanc} \xikj \, \Pkj \, \Q(\phiest_{kj},\delta_{kj})\Big)^{-1}\bigg\} \,.
	\end{align}
	Moreover, the tightest upper bound in \eqref{eq:wstbnd} is attained by
	\begin{align*}
		\sin\varepsilon^{\phi}_{kj} = \arg\min_{\delta_{kj}} \text{\tr}\bigg\{\Big(\sum_{j\in\Sanc} \xikj \, \Pkj \, \Q(\phiest_{kj},\delta_{kj})\Big)^{-1}\bigg\}.
	\end{align*}
\end{proposition}

\begin{proof}
	See Appendix \ref{apd:Qmatrix}.
\end{proof}

In the rest of the paper, we take the minimizer $\delta_{kj} = \sin\varepsilon^{\phi}_{kj}$ and denote the matrix 
\begin{align*}
	\Q(\phiest_{kj}) =  \Jr(\phiest_{kj}) - \sin\varepsilon^{\phi}_{kj}\cdot\I
\end{align*}
by omitting the variable $\delta_{kj}$ in \eqref{eq:matQ} for simplicity. Then, we replace the matrix $\Jr(\phi_{kj})$ with
$\Q(\phiest_{kj})$ in the previous formulation, and propose a robust counterpart of $\PI$ given by
\begin{align}
\PRI: \quad	
\min_{\{\Pkj\}} \quad& \sum_{k\in\Sagt} \tr\bigg\{\Big(\sum_{j\in\Sanc}  \xiwstkj \, \Pkj \, \Q(\phiest_{kj})\Big)^{-1}\bigg\}  \notag
	\\
	\text{s.t.} \quad
				& \sum_{j\in\Sanc} \xiwstkj \, \Pkj \, \Q(\phiest_{kj}) \succeq 0, \;\forall\, k\in\Sagt
				\label{eq:PRI-con-anc}
				\\
				& \eqref{eq:P-con-ttl}~\tx{--}~\eqref{eq:P-con-link}. \notag
\end{align}
Again by the property of Schur complement as in $\PISDP$, the problem $\PRI$ is equivalent to a SDP formulation, given by
\begin{align}
\PRISDP\!\!: 
\min_{\{\Pkj\},\,\M_k} \,& \sum_{k\in\Sagt} \tr\left\{\M_k\right\} 
	\notag\\
	\text{s.t.} \quad\,& \!\begin{bmatrix}
							\M_k & \I \\ 
							\I & \!\!\sum_{j\in\Sanc}  \xiwstkj \, \Pkj \, \Q(\phiest_{kj})
					   \end{bmatrix} \!\succeq\! 0, \,\forall k\in\Sagt
				\label{eq:PR-SDP-con-anc-psd}
				\\
				& \eqref{eq:P-con-ttl}~\tx{--}~\eqref{eq:P-con-link}.
				\notag
\end{align}

\begin{remark}
	The formulation with QoS guarantee proposed in Section \ref{sec:optimal-QoS} can also be extended to its robust formulation using the above method. By such, the SPEB of each agent is always guaranteed to satisfy its position error requirement. However, if using the non-robust formulation, the requirements for agents' SPEBs, e.g., \eqref{eq:con-speb} or \eqref{eq:con-mdpeb}, can easily be violated due to imperfect knowledge of network parameters.
\end{remark}

Note that from Proposition \ref{thm:Qmatrix}, the new formulation $\PRI$ is a valid relaxation for $\PR$ when the condition (\ref{eq:PRI-con-anc}) holds. Since $\Q(\phiest_{kj})$ is not positive definite due to $\det\big(\Q(\phiest_{kj})\big) = \sin\varepsilon^{\phi}_{kj} (\sin\varepsilon^{\phi}_{kj} - 1 ) \leq 0$, such a condition does not necessarily hold for all power allocation $\{\Pkj\}$. However, we will show that it holds for the optimal power allocation of $\PR$ with high probability (w.h.p.) when the number of anchors is large or the uncertainty in angle is small.

	Before giving the proposition, we introduce an equivalent expression for the channel parameter $\xikj$ in \eqref{eq:ch-para} as $\xikj = \zeta_{kj}/d_{kj}^{\,2\beta}$, where $\zeta_{kj}$ is a positive coefficient characterizing shadowing effect and small-scale fading process, and $\beta$ is the amplitude loss exponent.\footnote{We introduce the path loss model here to facilitate the proof of the Proposition \ref{thm:prbound}. However, the robust power allocation schemes do not require $\beta$, since the channel parameter $\xikj$ can be obtained directly through channel estimation.}

\begin{proposition}\label{thm:prbound}
Consider a network where all the nodes are uniformly located in a $R \times R$ square region, the minimum distance between two nodes is $r_0$, and the coefficient $\zeta_{kj}$ has a support on $[\zeta_{\min}~\zeta_{\max}]$ where $0 < \zeta_{\min} \leq \zeta_{\max}$. Let $\{\Pkj^*\}$ be the optimal solution of $\PR$, and $\delta = \sin\varepsilon^{\phi}$ where $\varepsilon^{\phi} = \max\{\varepsilon^{\phi}_{kj}\}$, then
	\begin{enumerate}[(a)]
		\item when $\Nanc\rightarrow\infty$ and $\delta\leq\delta_{\max}$, where $\delta_{\max}$ is the smallest positive root of equation
$4 \delta^4 - 4 \delta^2 - 2 {\zeta_{\max}}/{\zeta_{\min}} \delta + 1 = 0$,
we have
			\begin{align*}
				\Pr\bigg\{\sum_{j\in\Sanc}  
				\xiwstkj \, \Pkj^* \, \Q(\phiest_{kj}) \succeq 0 \bigg\} 
				= 1 - \mathcal{O}\big(\exp(-\eta\!\cdot\!\Nanc)\big), 
\\ \quad \forall k\in\Sagt
			\end{align*}
			where $\eta$ is a fixed positive number;
		\item when $\varepsilon^{\phi}\rightarrow 0$, we have
			\begin{align*}
				\Pr\bigg\{\sum_{j\in\Sanc}  
				\xiwstkj \, \Pkj^* \, \Q(\phiest_{kj}) \succeq 0\bigg\} 
					= 1 - \mathcal{O}\big((\varepsilon^\phi)^{\Nanc/2}\big)\,, 
\\ \quad \forall k\in\Sagt.
			\end{align*}
	\end{enumerate}
\end{proposition}

\begin{proof}
	See Appendix \ref{apd:prbound}.
\end{proof}

\begin{remark}
	Proposition \ref{thm:prbound} implies that the condition (\ref{eq:PRI-con-anc}) holds w.h.p. at the rate indicated by the $\mathcal{O}$ notation, where $\mathcal{O}(f(n))$ means that the function value is on the order of $f(n)$ \cite{CLRS:01}.
\end{remark}

\begin{remark}
	Note that Proposition \ref{thm:prbound} holds for $\{\Pkj^*\}$, which implies that the optimal solution of the original robust formulation $\PR$ is included in the feasible set of the proposed formulation $\PRI$ (or $\PRISDP$) w.h.p.
\end{remark}

\subsection{Robust Counterpart of mDPEB Minimization}
We investigate the robust power allocation based on mDPEB formulation $\PII$. To circumvent the intractable maximization in \eqref{eq:maxphi}, we consider the robust SPEB formulation $\PRI$. Specifically, the objective of $\PRI$ can be written as
\begin{align}\label{eq:SPEB-robust}
	\tr\bigg\{\Big(\sum_{j\in\Sanc}  \xiwstkj \, \Pkj \, \Q(\phiest_{kj})\Big)^{-1}\bigg\} 
	= \frac{1}{\muwst_{1,k}} + \frac{1}{\muwst_{2,k}}
\end{align}
where $\muwst_{1,k}$ and $\muwst_{2,k}$ are the two eigenvalues of the matrix
$\sum_{j\in\Sanc}  \xiwstkj \, \Pkj \, \Q(\phiest_{kj})$, 
given by
\begin{align}\label{eq:mu-robust}
	\muwst_{1,k}, ~\muwst_{2,k}
	=\frac{1}{2}\bigg( &\sum_{j\in\Sanc} \xiwstkj \, \Pkj (1 - 2\sin\varepsilon^{\phi}_{kj})
		\notag\\
		&\pm \Big\|
			\sum_{j\in\Sanc} \xiwstkj \, \Pkj \,\uvec(2\hat{\phi}_{kj}) \Big\| \bigg).
\end{align}
Geometrically, $\muwst_{1,k}$ and $\muwst_{2,k}$ are similar to the DPEB's in two orthogonal directions. 
Using Proposition \ref{thm:prbound}, we can show that $\muwst_{2,k}\geq 0$ w.h.p. when $\Nanc$ is large or $\varepsilon^{\phi}$ is small. Since $\muwst_{1,k}\geq\muwst_{2,k}$, the smaller eigenvalue $\muwst_{2,k}$ dominates the function in \eqref{eq:SPEB-robust}. Hence, we formulate a robust counterpart of $\PII$ based on $\muwst_{2,k}$, given by
\begin{align}
	\PRII: \quad	
	\min_{\{\Pkj\}} \quad& \sum_{k\in\Sagt} \frac{1}{\muwst_{2,k}}
	\notag\\
	\text{s.t.}	\quad
				& \muwst_{2,k} \geq 0, \quad \forall k\in\Sagt
				\label{eq:PRII-con-anc}
				\\
				& \eqref{eq:P-con-ttl}~\tx{--}~\eqref{eq:P-con-link}.
				\notag
\end{align}
Given that $\muwst_{2,k}\geq 0$, the problem $\PRII$ is equivalent to the following SOCP problem:
\begin{align}
	\PRIICP\!\!: 
	\min_{\{\Pkj,r_k\}} ~& \sum_{k\in\Sagt} 
	\frac{1}{\sum_{j\in\Sanc} \xiwstkj \, \Pkj \big(1 - 2\sin\varepsilon^{\phi}_{kj}\big) - r_k}
	\label{eq:PmDR-obj}\\
	\text{s.t.} \quad& r_k \geq 
				\Big\|\sum_{j\in\Sanc} \xiwstkj \, \Pkj \, \uvec(2\phiest_{kj})
				\Big\|, \quad \forall k\in\Sagt
				\label{eq:PmDR-con-soc}\\
				& r_k \leq 
\sum_{j\in\Sanc} \xiwstkj \, \Pkj 
\big(1 \!-\! 2\sin\varepsilon^{\phi}_{kj}\big), ~\forall k\in\Sagt
				\notag\\
				& \eqref{eq:P-con-ttl}~\tx{--}~\eqref{eq:P-con-link}.
				\notag
\end{align}
Note that the uncertainty in angle $\varepsilon^{\phi}_{kj}$ only exists in the objective, and does not affect the second-order conic constraint \eqref{eq:PmDR-con-soc}. Hence, the problem $\PRIICP$ retains the same structure of $\PIICP$, and its optimal solution can be efficiently obtained.


\section{Efficient Robust Algorithm Using \\Distributed Computations}\label{sec:distributed}
In this section, we designed a distributed robust algorithm for both SPEB and mDPEB minimization, which decomposes the original formulation into two-stage optimization problems and enables parallel computations among all the agents. The proposed algorithms achieve the global optimal solution with improved computational efficiency.

\subsection{Algorithm for SPEB Minimization}
Despite the convexity of the robust SDP formulation $\PRISDP$, there are multiple positive semidefinite constraints imposed for multiple agents, and the computational complexity depends on the number of SDP constraints. 
To efficiently obtain the power allocation decision for multi-agent networks, we design a distributed implementation for $\PRISDP$, which can be solved using parallel computations among the agents.

Specifically, we let $\Pkj = \rhokj\Pk$ where $\Pk$ is the total power assigned for locating agent $k$, and $\rhokj\in[0,1]$ is a fractional number denoting the percentage of $\Pk$ allocated to anchor $j$. By introducing the two variables $\rhokj$ and $\Pk$, the robust formulation for power allocation can be written as
\begin{align}
\min_{\{\rhokj,\Pk\}} ~& \sum_{k\in\Sagt} \frac{1}{\Pk} \tr\bigg\{\Big(\sum_{j\in\Sanc}  \xiwstkj \, \rhokj \, \Q(\phiest_{kj})\Big)^{-1}\bigg\} 
	\notag\\
	\text{s.t.} \quad~
			& \sum_{j\in\Sanc} \rhokj \leq 1 \label{eq:PR-I-con-ttl}
			\\
			& \rhokj \geq 0, \quad \forall k\in\Sagt, ~\forall j\in\Sanc
			\label{eq:PR-I-con-anc} \\
			& \sum_{k\in\Sagt} \Pk \leq \Ptot
			\label{eq:PR-I-con-pow} \\
			& \Pk \geq 0, \quad \forall k\in\Sagt.
			\label{eq:PR-I-con-agt}
\end{align}
Since the constraints on $\rhokj$ and $\Pk$ are separable, and $\Pk$ and $\rhokj$ are only related to the SPEB of agent $k$, we can decompose the above problem into two stages. In Stage I, given the total power budget $\Pk$ for agent $k$, we consider the optimal allocation of $\Pk$ among all the anchors, i.e.,
\begin{align*}
\PRIDI: \quad	
\min_{\{\rhokj\},\,\M_k} ~& \tr\left\{\M_k\right\} / \Pk
	\notag\\
	\text{s.t.} \quad~& \begin{bmatrix}
							\M_k & \I \\ 
							\I & \sum_{j\in\Sanc}  \xiwstkj \, \rhokj \, \Q(\phiest_{kj})
					   \end{bmatrix} \succeq 0 
				\notag\\
				& \eqref{eq:PR-I-con-ttl}~\tx{--}~\eqref{eq:PR-I-con-anc}.
\end{align*}
The optimal solution of $\PRIDI$ is denoted by $\rhokj^*$, and it is independent of the total power for agent $k$ since $\Pk$ only appears as a scaler in the objective and can be removed. Since the problem $\PRIDI$ is formulated for agent $k$, there are totally $\Nagt$ problems to be solved in Stage I.

In Stage II, we allocate the total $\Pk$ for localizing agent $k$. The objective is the total SPEB of the agents, where the parameter $\rhokj^*$'s are from Stage I $\PRIDI$. In particular, we let $T_k = \tr\big\{\big(\sum_{j\in\Sanc} \xiwstkj \, \rhokj^* \, \Q(\phiest_{kj})\big)^{-1}\big\}$ and formulate the problem as:
\begin{align*}
\PRIDII: \quad	
\min_{\{\Pk\}} \quad& \sum_{k\in\Sagt} \frac{T_k}{\Pk} 
	\notag\\
	\text{s.t.} \quad 
	& \eqref{eq:PR-I-con-pow}~\tx{--}~\eqref{eq:PR-I-con-agt}. 
\end{align*}
The problem $\PRIDII$ is convex in $\Pk$, and the optimal solution is given in a closed form as follows.
\begin{proposition}\label{thm:Pk-close}
	Given that $\rhokj^*$ is the optimal solution of $\PRIDI$, the optimal solution of $\PRIDII$ is given by
	\begin{align}\label{eq:Pk-close}
		\Pk^* = \frac{\Ptot\sqrt{T_k}} 
					  {\sum_{k\in\Sagt} \sqrt{T_k}} \,.
	\end{align}
\end{proposition}
\begin{proof}
	See Appendix \ref{apd:Pk-close}.
\end{proof}

The optimal power allocation for the location-aware network is 
\begin{align}\label{eq:optsolu}
	\Pkj^* = \rhokj^*\Pk^*
\end{align}
where $\Pk^*$ is given in \eqref{eq:Pk-close}.
The detailed algorithm is described in the Algorithm \ref{alg:m-agent}.

\begin{algorithm}[!h]
	\caption{\label{alg:m-agent}Robust power allocation algorithm for multiple-agent networks}
\begin{algorithmic}[1]
	\Require the angle $\phiest_{kj}$ and the distance $\dest_{kj}$ between anchor $j$ ($j\in\Sanc$) and agent $k$ ($k\in\Sagt$)
	\State Set $\Pk\gets 1$, $\forall k\in\Sagt$
	\State Solve the Stage I problems $\PRIDI$ which gives the optimal solution $\rho_{kj}^*$
	\State 
	Set $\rho_{kj}\gets \rho_{kj}^*$, $\forall k\in\Sagt$, $\forall j\in\Sanc$
	\State
	Solve the Stage II problem $\PRIDII$ by using \eqref{eq:Pk-close} to compute the optimal solution $\Pk^*$
	\State Set $\Pkj^*\gets \rhokj^* \Pk^*$, $\forall k\in\Sagt$, $\forall j\in\Sanc$
\end{algorithmic}
\end{algorithm}

\begin{remark}
Since each Stage I problem $\PRIDI$ in Algorithm \ref{alg:m-agent} is with a single SDP constraint, its complexity is much lower than the original problem $\PRISDP$ which contains $\Nagt$ SDP constraints.
Moreover, the $\Nagt$ Stage I problems $\PRIDI$ can be separately solved by the $\Nagt$ agents, since each agent itself does not require any information from other agents. Thus, the computation efficiency can be improved by $\Nagt$ times using the parallel computations among the agents.
\end{remark}

\begin{remark}
	The proposed distributed algorithm can also be applied to the robust power allocation with individual power constraint, e.g., $\sum_{k\in\Sagt} \Pkj \leq \Pjtot$. In particular, we replace such constraint with $\sum_{k\in\Sagt} \rhokj\Pk \leq \Pjtot$
in the Stage II formulation $\PRIDII$, while the Stage I formulation $\PRIDI$ remains the same. In such case, the close-form solution in \eqref{eq:optsolu} is not available, however, the optimal solution of the Stage II problem can still be efficiently obtained since the problem is convex. Consequently, we can obtain a sub-optimal solution for the overall problem. 
\end{remark}

\begin{figure}[t!]
	\centering
	\psfrag{Anchor}[][]{\hspace{-0.5em}\scriptsize{Anchor}}
	\psfrag{Agent}[][]{\hspace{-0.5em}\scriptsize{Agent}}
	\includegraphics[height=7cm]{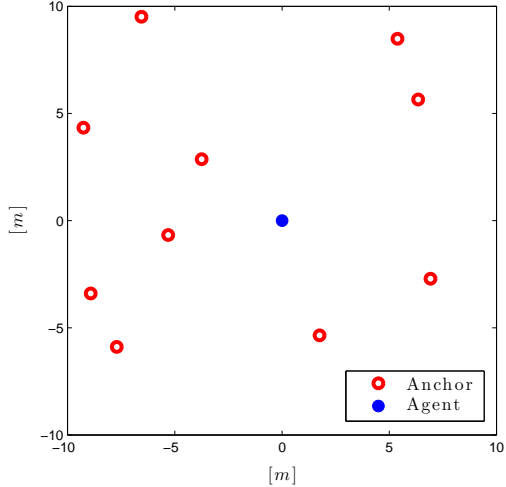}
	\caption{\label{fig:topo-single}The location-aware network consisting ten anchors (red circle) and one agents (blue dot), where the anchors are uniformly distributed in the square region.}
\end{figure}

\subsection{Algorithm for mDPEB Minimization}
A similar decomposition method can be applied to the mDPEB minimization $\PRII$, i.e., by introducing two variables $\rhokj$ and $\Pk$. Instead of solving SDP in SPEB minimization, each agent will separately solve a SOCP problem with linear objective for the mDPEB minimization. Specifically, we rewrite \eqref{eq:mu-robust} as
\begin{align*}
	\muwst_{2,k}
	=\frac{\Pk}{2}\bigg( &\sum_{j\in\Sanc} \xiwstkj \, \rho_{kj} \left(1 - 2\sin\varepsilon^{\phi}_{kj}\right)
	\notag\\
	&- \Big\| \sum_{j\in\Sanc} \xiwstkj \, \rho_{kj} \,\uvec(2\hat{\phi}_{kj}) \Big\| \bigg).
\end{align*}
Then, the two-stage formulations are given by
\begin{align*}
	\PRIIDI: \quad	
	\max_{\{\rho_{kj}\}} \quad& \muwst_{2,k}/\Pk
	\\
	\text{s.t.} \quad& \muwst_{2,k} \geq 0
				\notag\\
				& \eqref{eq:PR-I-con-ttl}~\tx{--}~\eqref{eq:PR-I-con-anc}
\end{align*}
and
\begin{align*}
	\PRIIDII: \quad	
	\min_{\{\Pk\}} \quad& \sum_{k\in\Sagt} \frac{1}{\muwst_{2,k}} \\
	\text{s.t.} \quad&
	\eqref{eq:PR-I-con-pow}~\tx{--}~\eqref{eq:PR-I-con-agt}
\end{align*}
respectively. The optimal power allocation is the product of the optimal solutions of the two-stage problems, given by \eqref{eq:optsolu}. The algorithm for mDPEB minimization is similar to that of Algorithm \ref{alg:m-agent}, and hence, we omit the details here.

\section{Simulation Results}\label{sec:simu}

In this section, we investigate the localization performance by the proposed power allocation schemes. The total power for localization is normalized to $\Ptot = 1$, and the channel parameter is given by $\xikj=10^3/d_{kj}^2$.\footnote{We choose the free-space propagation model where the path loss exponent is 2 \cite{MolGreSha:09}.}
The proposed optimization of power allocation, i.e., SDP and SOCP, are solved by the standard optimization solver CVX \cite{cvx}.

\begin{figure}[t!]
	\centering
	\includegraphics[height=7cm]{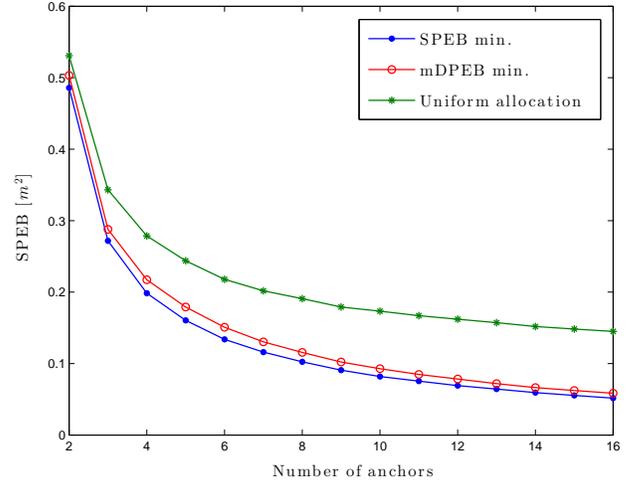}
	\caption{\label{fig:nr-speb-single}The SPEB in single-agent networks with respect to the number of anchors, obtained by different power allocation schemes.}
\end{figure}

\subsection{Power Allocation with Perfect Network Parameters}

First, we investigate the SPEB with power allocation as the number of anchors or agents changes. Three schemes of power allocation are compared: the allocation via SPEB minimization formulated in $\PISDP$, the allocation via mDPEB minimization formulated in $\PIICP$, and the uniform allocation which assigns $\Ptot$ equally over all the anchors. Given the number of anchors and agents, we run Monte Carlo simulation to generate $10^3$ deployments of agents or anchors that are uniformly distributed in a squared region, i.e., $U(\,[-10,10]\times[-10,10]\,)$, and then compute the average SPEB obtained by each scheme. 

In Figs.~\ref{fig:topo-single} and \ref{fig:nr-speb-single}, we consider the network with a single agent at the center and anchors uniformly distributed. An example of the network topology is illustrated in Fig. \ref{fig:topo-single}. We plot the SPEBs obtained by the above-mentioned three schemes in Fig.~\ref{fig:nr-speb-single}. A decreasing tendency in SPEB is observed as the number of anchors increases. This is reasonable since the agent has more freedom to choose ``good'' anchors when there are more anchors. Moreover, the results show that the mDPEB minimization outperforms the uniform allocation by about $46\%$, and achieves a SPEB close to the one obtained by SPEB minimization.

Next, we consider a network with multiple agents. Ten anchors are placed with fixed locations, and the agents are uniformly distributed in the region (see Fig.~\ref{fig:topo-multiple}). Similarly, we compare the SPEB obtained by the three schemes with respect to the number of agents in Fig. \ref{fig:nr-speb-multiple}. It shows that, even in multiple-agent case, the mDPEB minimization still achieves a similar performance as the SPEB minimization, and remarkably outperforms the uniform allocation. It implies that mDPEB is a meaningful performance metric for the optimization of power allocation. In addition, we observe that the average SPEB increases linearly with the number of agents. This is because each agent tends to obtain less power when the total power budget is fixed. As indicated by the slope, the speed of SPEB increase of optimized allocation is about $60\%$ slower than that of uniform allocation.

Furthermore, we investigate the performance of the two-stage optimization proposed in Section \ref{sec:distributed} which exploits the distributed computations among multiple agents. In Fig.~\ref{fig:nr-speb-multiple}, we plot the SPEB obtained by the two-stage optimization for both SPEB and mDPEB minimization. The results show that the SPEB solved by two-stage optimization perfectly matches that of one-stage optimization, which validates that the two-stage scheme can obtain the optimal solution while requiring much less computational time. 

\begin{figure}[t!]
	\centering
	\includegraphics[height=7cm]{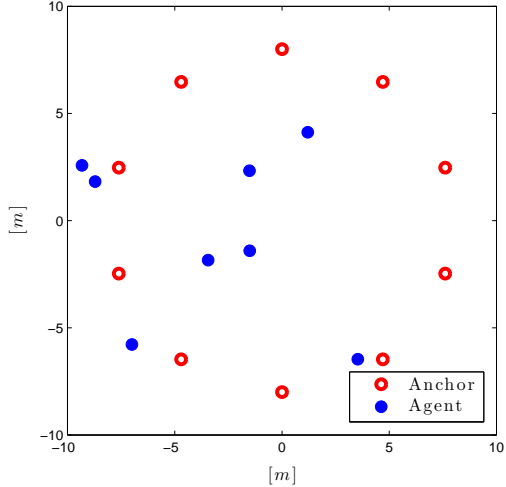}
	\caption{\label{fig:topo-multiple}The location-aware network consisting ten anchors (red circle) and eight agents (blue dot), where the agents are uniformly distributed in the square region.}
\end{figure}

\subsection{Robust Power Allocation with Imperfect Knowledge of Network Parameters}

We then investigate the performance of the power allocation with imperfect knowledge of network parameters. We compared the following schemes: allocation by the robust formulation $\PRISDP$ and $\PRIICP$, allocation by the non-robust formulation $\PISDP$ and $\PIICP$, and uniform allocation. 
We consider the agent's actual position lies within a circle of radius $\varepsilon^d$ centering at its estimated position. Then the maximum angular uncertainty is determined by $\varepsilon^{\phi}_{kj} = \arcsin(\varepsilon^d/\hat{d}_{kj})$.\footnote{Without loss of generality, we set $\varepsilon^d_{kj} = \varepsilon^d$ for all $k$, $j$.} The \emph{normalized uncertainty set size} on network parameters is defined to be $\varepsilon = 2\varepsilon^d/20$ which is normalized by the length of the squared region.

In Fig.~\ref{fig:rb-speb-nanc}, we investigate the actual SPEB with respect to the number of anchors. We consider a single-agent network, and set the normalized uncertainty set size $\varepsilon$ to be $0.2$, i.e., $\varepsilon^d = 2$ m. The results show that the robust SPEB minimization ($\PRISDP$) outperforms the non-robust SPEB minimization ($\PISDP$) by $20\%$, and outperforms uniform allocation by $35\%$; the robust mDPEB minimization ($\PRIICP$) outperforms the non-robust mDPEB minimization ($\PIICP$) by $30\%$, and outperforms uniform allocation by $70\%$. Moreover, we observe that the actual SPEB of robust mDPEB minimization is smaller than that of robust SPEB minimization, and the same observation is on the non-robust schemes. It implies that the mDPEB minimization is more robust to the network parameter uncertainty, compared with the SPEB minimization. This can be explained as follows: the robust mDPEB minimization can be viewed as a doubly robust optimization, since it first minimizes the maximum positional error over all the directions. Therefore, $\PRIICP$ outperforms $\PRISDP$ when the uncertainty in network parameters is not negligible (e.g., $\varepsilon = 0.2$).

\begin{figure}[t!]
	\centering
	\includegraphics[height=7cm]{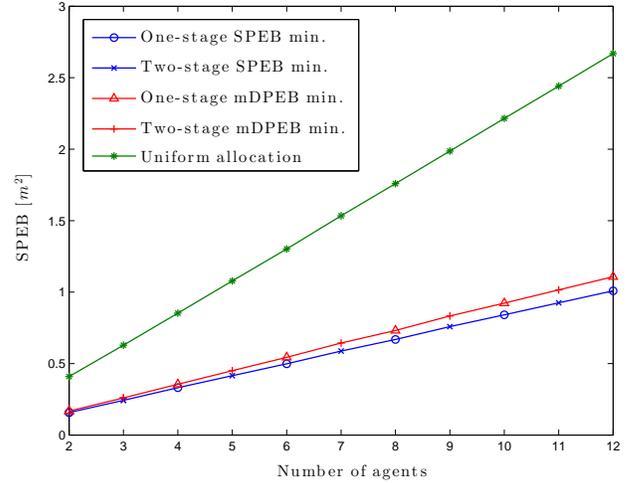}
	\caption{\label{fig:nr-speb-multiple}The average SPEB in multiple-agent networks ($\Nanc\!=\!10$) by different power allocation schemes.
Both one-stage and two-stage optimization are considered.
}
\end{figure}

In Fig.~\ref{fig:rb-speb-err}, we investigate the actual SPEB with respect to the normalized uncertainty set size $\varepsilon$. We consider a single-agent network with ten anchors deployed on a circle (similar to Fig. \ref{fig:topo-multiple}). As we observe, the actual SPEB of non-robust schemes quickly increases as the normalized uncertainty set size goes large. When the normalized uncertainty set size is larger than $0.22$ and $0.27$, respectively, the non-robust SPEB minimization and non-robust mDPEB minimization even perform worse than the uniform allocation, while the robust schemes always achieves better SPEB than all the other schemes. Moreover, the robust mDPEB minimization outperforms the non-robust mDPEB minimization and robust SPEB minimization by $30\%$ and $23\%$, respectively, when $\varepsilon=0.15$. Both Figs.~\ref{fig:rb-speb-nanc} and \ref{fig:rb-speb-err} have demonstrated the advantage of the proposed robust power allocation schemes, especially the mDPEB minimization, in the practical location-aware networks with imperfect knowledge of network parameters.

\begin{figure}[t!]
	\centering
	\includegraphics[height=7cm]{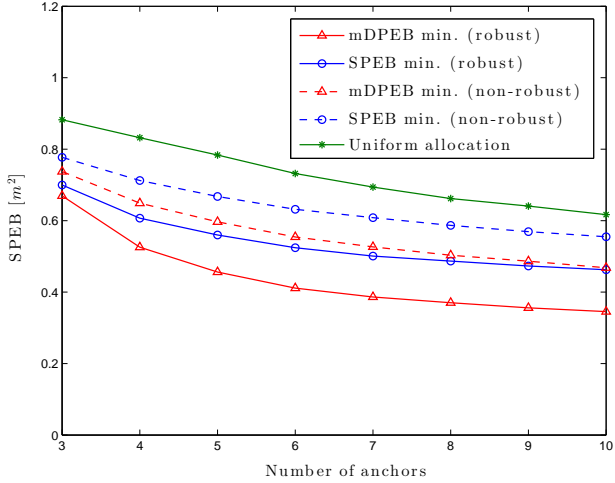}
	\caption{\label{fig:rb-speb-nanc} The actual SPEB with respect to number of anchors, obtained by different power allocation schemes
with imperfect knowledge of network parameters ($\varepsilon = 0.2$).}
\end{figure}

\section{Conclusion}\label{sec:conclude}

In this paper, we presented an optimization framework for robust power allocation in network localization based on the performance metrics SPEB and mDPEB. We first showed that the optimal power allocation with perfect network parameters can be efficiently obtained via conic programming, and then proposed robust power allocation schemes to combat uncertainty in network parameters for practical systems. Moreover, we designed an efficient algorithm for robust power allocation that allows distributed computations among agents. The simulation results demonstrated that the robust power allocation remarkably outperforms the non-robust power allocation and uniform allocation. Furthermore, we showed that, compared with the SPEB minimization, the mDPEB minimization is more robust to network parameter uncertainty for power allocation.

\appendices

\section{Proof of Proposition \ref{thm:mDPEB}}\label{apd:mDPEB}

The maximization on DPEB in \eqref{eq:mDPEB} follows that:
\begin{align}
	&\hspace{-1.5em} \max_{\varphi\in[0,2\pi)} 
	\left\{\mathcal{P}(\mathbf{p}_k;\varphi)\right\} 
	\notag \\
	=& \max_{\varphi\in[0,2\pi)}~
	\uvec(\varphi)\T [\Jeinv(\mathbf{p}_k;\{\Pkj\})] \uvec(\varphi) 
	\notag \\
	=& \max_{\varphi\in[0,2\pi)}~
	\uvec(\varphi)\T (\mathbf{U}_{\theta_k}^{-1})\T
	\begin{bmatrix} \mu_{1,k}^{-1} & 0\\ 0 & \mu_{2,k}^{-1} \end{bmatrix}
	\mathbf{U}_{\theta_k}^{-1} \uvec(\varphi)
	\notag \\
	=& \max_{\varphi'\in[0,2\pi)}~
	\uvec(\varphi')\T [\Jeinv(\mathbf{p}_k;\{\Pkj\})] \uvec(\varphi')
	\label{eq:maxDPEB}
\end{align}
where the last equality is due to the fact that the product of a unit vector and a rotation matrix $\mathbf{U}_{\theta_{k}}$ is still a unit vector. Now, let $\varphi'=\theta_k$ in \eqref{eq:maxDPEB}, then we have
\begin{align*}
	\max_{\varphi\in[0,2\pi)} 
	\left\{\mathcal{P}(\mathbf{p}_k;\varphi)\right\} 
	= \,& \max_{\theta_k} \left\{\mu_{1,k}^{-1}\cos^2\theta_k + \mu_{2,k}^{-1}\sin^2\theta_k\right\}
	\notag \\
	= \,&\, \mu_{2,k}^{-1}
\end{align*}
where the last equation is due to $\mu_{1,k} \geq \mu_{2,k}$.

%
%

\begin{figure}[t!]
	\centering
	\includegraphics[height=7cm]{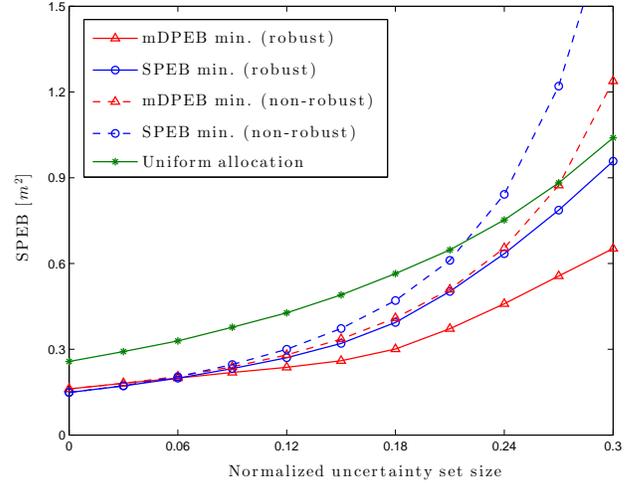}
	\caption{\label{fig:rb-speb-err}The actual SPEB with respect to the normalized uncertainty set size on network parameters, obtained by different power allocation schemes.
	}
\end{figure}

\section{Proof of Proposition \ref{thm:convexity}}\label{apd:convexity}

Since \eqref{eq:P-con-ttl}--\eqref{eq:P-con-link} are all linear constraints, we only need to show the objective in \eqref{eq:P-obj}, i.e., the SPEB, is a convex function in $\Pkj$. We write the transmit power of agent $k$ as a vector $\mathbf{x}_k = [x_{k1}\;x_{k2}\;\cdots\;x_{k\Nanc}]\T$, and the SPEB is a function of $\mathbf{x}_k$, given by
\begin{equation*}
	f(\mathbf{x}_k) \triangleq 
	\tr\bigg\{\Big(\sum_{j\in\Sanc} \xikj \, \Pkj \, \Jr(\phi_{kj})\Big)^{-1}\bigg\}.
\end{equation*}
We choose two arbitrary  $\mathbf{x}_k,~\mathbf{x}'_k\in\mathbb{R}^{\Nanc}_+$. Given any $\alpha\in[0,1]$, we have 
\begin{align}
	&\hspace{-1em} f(\alpha\mathbf{x}_k + (1-\alpha)\mathbf{x}'_k)
	\notag \\ 
	=&~ \tr{\bigg\{}{\Big(} \sum_{j\in\Sanc} \xikj \left(\alpha \Pkj + (1-\alpha) \Pkj' \right) \Jr(\phi_{kj}) {\Big)}^{-1}{\bigg\}} 
	\notag \\
	=&~ \tr\bigg\{\!\Big( \!\alpha \!\! \sum_{j\in\Sanc} \! \xikj \, \Pkj \, \Jr(\phi_{kj})  + (1\!-\alpha) \!\! \sum_{j\in\Sanc} \! \xikj \, \Pkj' \, \Jr(\phi_{kj}) \!\Big)^{\!\!-1}\!\bigg\} 
	\notag \\
	\leq&~ \alpha f(\mathbf{x}_k) + (1-\alpha) f(\mathbf{x}'_k). \label{eq:trcvx}
\end{align}
The inequality \eqref{eq:trcvx} holds since the function $\tr\left\{\mathbf{X}^{-1}\right\}$ is convex in $\mathbf{X}\succ 0$ \cite{BoyVan:B04}. If the matrix $\mathbf{X}$ is singular, the inequality \eqref{eq:trcvx} still holds. Since $\xikj$ is a positive scaler, $f(\mathbf{x}_k)$ is convex in $\mathbf{x}_k$.

\section{Proof of Proposition \ref{thm:Qmatrix}}\label{apd:Qmatrix}

Let $\phi_{kj}^+ = \phi_{kj}+\phiest_{kj}$ and 
$\phi_{kj}^- = \phi_{kj}-\phiest_{kj}$, we have
\begin{align*}
	&\hspace{-1em} \Jr(\phi_{kj}) - \Q(\phiest_{kj},\delta_{kj}) 
	\notag\\
	=& \begin{bmatrix}
		\delta_{kj} - \sin\phi_{kj}^+ \, \sin\phi_{kj}^- 
		& \cos\phi_{kj}^+ \, \sin\phi_{kj}^- \\
		\cos\phi_{kj}^+ \, \sin\phi_{kj}^- 
		& \delta_{kj} + \sin\phi_{kj}^+ \, \sin\phi_{kj}^-
	\end{bmatrix}. 
\end{align*}
We can show that $\Jr(\phi_{kj}) - \Q(\phiest_{kj},\delta_{kj})$ is positive semidefinite if
\begin{align*}
	\begin{cases}
		\delta_{kj} \geq \sin\phi_{kj}^+ \, \sin\phi_{kj}^-\,, \\
		\delta_{kj} \geq |\sin\phi_{kj}^-|\,.
	\end{cases}		
\end{align*}
Since $|\phi_{kj}^-| \leq \varepsilon^{\phi}_{kj}$, the above two inequality conditions are guaranteed by
\begin{align*}
	\delta_{kj} \geq \sin\varepsilon^{\phi}_{kj}\,.
\end{align*}
Given that $\sum_{j\in\Sanc} \xikj \, \Pkj \, \Q(\phiest_{kj},\delta_{kj}) \succeq 0$, we have 
\begin{align*}
	&\text{\tr}\bigg\{\Big(\sum_{j\in\Sanc} \xikj \, \Pkj \, \Jr(\phi_{kj})\Big)^{-1}\bigg\} 
	\notag\\
	&\leq 
	\text{\tr}\bigg\{\Big(\sum_{j\in\Sanc} \xikj \, \Pkj \, \Q(\phiest_{kj},\delta_{kj})\Big)^{-1}\bigg\}	
\end{align*}
for all $\phi_{kj}\in\mathcal{S}^{\phi}_{kj}$.
Furthermore, we can show that $\Q(\phiest_{kj},\delta_1) \preceq \Q(\phiest_{kj},\delta_2)$ for $0\leq\delta_2\leq\delta_1$, which implies that the function $\tr\big\{\big(\sum_{j\in\Sanc} \xikj \, \Pkj \, \Q(\phiest_{kj},\delta_{kj})\big)^{-1}\big\}$ is a non-decreasing function of $\delta_{kj}$. Hence, the minimum value of the right-hand side of \eqref{eq:wstbnd} is obtained when $\delta_{kj} = \sin\varepsilon^{\phi}_{kj}$.


\section{Proof of Proposition \ref{thm:prbound}}\label{apd:prbound}

We first consider the network with a single agent, and then extend the proof to the multiple-agent case. 
For a given $k\in\Sagt$, we need to show that the condition (\ref{eq:PRI-con-anc}) holds
for $\{\Pkj^*\}$ w.h.p. for both cases (a) and (b). Note that since
\begin{align*}
	\sum_{j\in\Sanc} \xiwstkj \, \Pkj^* \, \Q(\phiest_{kj}) \succeq\! \sum_{j\in\Sanc} \xiwstkj \, \Pkj^* \, \Jr(\phiest_{kj}) -  \frac{\zeta_{\max}}{r_0^{2\beta}} \Ptot \, \delta_{kj} \, \I
\end{align*}
it is sufficient to show that w.h.p.
\begin{align}\label{eq:pfcond}
	\tr \bigg\{ \Big( \sum_{j\in\Sanc} \xiwstkj \, \Pkj^* \, \Jr(\phiest_{kj})\Big)^{-1} \bigg\} \leq \frac{r_0^{2\beta}}{\zeta_{\max}} \frac{2}{\Ptot \, \delta}
\end{align}
where $\delta = \sin\varepsilon^{\phi}$ with $\varepsilon^{\phi} = \max\{\varepsilon^{\phi}_{kj}\}$.

\begin{figure}[t!]
	\centering
	\psfrag{k}[][]{\small{$k$}}
	\psfrag{i}[][]{\small{$i$}}
	\psfrag{i'}[][]{\small{$i'$}}
	\psfrag{phi}[][]{\small{$\Delta^{\phi}$}}
	\psfrag{r}[][]{\small{$r_0$}}
	\psfrag{br}[][]{\small{$\varrho r_0$}}
	\psfrag{R*R}[][]{\small{$R\times R$}}
	\vspace{0.5em}
	\includegraphics[height=5.5cm]{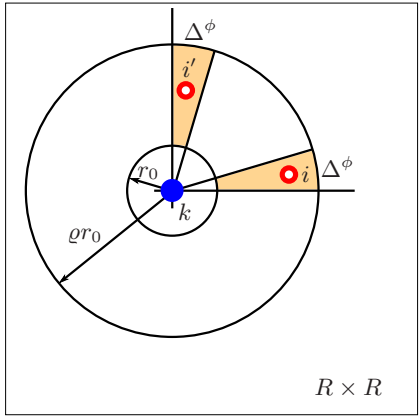}
	\caption{\label{fig:proofa}Geometrical illustration of the proof of Proposition \ref{thm:prbound}(a) where agent is inside the square region. We choose two anchors $i$ and $i'$ in the shaded region.}
\end{figure}


For (a): we pick two anchors $i$ and $i'$ in the region (see Fig. \ref{fig:proofa}) such that 
\begin{enumerate}
	\item $r_0 \leq \dwst_{ki},\dwst_{ki'} \leq \varrho r_0$ with $\varrho > 1$;
	\item $0\leq \phi_{ki} \leq \Delta^{\phi}$ and ${\pi}/{2} - {\Delta^{\phi}}\leq \phi_{ki'} \leq {\pi}/{2}$ for a small positive ${\Delta^{\phi}}$. 
\end{enumerate}
Note that if the agent is at the corner or on the boundary of the square area, we can rotate the angles accordingly to find such a region.

It can be shown that there exists at least one such pair of anchors with probability 
$1+(1-2p_0)^{\Nanc}-2(1-p_0)^{\Nanc}$, 
where 
$p_0 = (\varrho^2-1) r_0^2 \Delta^{\phi}/2R^2$.
Since the probability goes to 1 exponentially with $\Nanc$, such a pair of anchors can be found w.h.p.

Consider a power allocation scheme $\{\breve{P}_{ki}=\breve{P}_{ki'}= \Ptot/2\}$, and we show this scheme satisfies the condition \eqref{eq:pfcond} for a sufficiently small $\delta$. Based on the definition of the optimal power allocation, we have
\begin{align*}
	&\hspace{-1em} \tr \bigg\{ \Big( \sum_{j\in\Sanc} \xiwstkj \, \Pkj^* \, \Jr(\phiest_{kj})\Big)^{-1} \bigg\} \\
	\leq \,& \max_{\{\phi_{kj}\}} \tr \bigg\{ \Big( \sum_{j\in\Sanc} \xiwstkj \, \breve{P}_{kj} \, \Jr(\phi_{kj})\Big)^{-1} \bigg\} \\
	\leq \,& \max_{\{\phi_{kj}\}} \tr \bigg\{ \Big( \frac{\zeta_{\min}}{\varrho^{2\beta} r_0^{2\beta}} \frac{\Ptot}{2} \left(\Jr(\phi_{ki}) + \Jr(\phi_{ki'}) \right) \Big)^{-1} \bigg\} \\
	= \,&\, \frac{\varrho^{2\beta} r_0^{2\beta}}{\zeta_{\min}} \frac{2}{\Ptot} \frac{2}{\sin^2(\pi/2-2\Delta^{\phi}-2\varepsilon^{\phi})} \,.
\end{align*}
Therefore, a sufficient condition for \eqref{eq:pfcond} is
\begin{align*}
	\frac{\varrho^{2\beta} r_0^{2\beta}}{\zeta_{\min}} \frac{2}{\Ptot} \frac{2}{\sin^2(\pi/2-2\Delta^{\phi}-2\varepsilon^{\phi})} 
	\leq \frac{r_0^{2\beta}}{\zeta_{\max}} \frac{2}{\Ptot \, \delta}
\end{align*}
which is equivalent to
\begin{align}\label{eq:pfconda}
	\frac{2\varrho^{2\beta}\sin\varepsilon^{\phi}}{\cos^2(2\Delta^{\phi}+2\varepsilon^{\phi})} \leq \frac{\zeta_{\min}}{\zeta_{\max}}
\end{align}
where $\delta = \sin\varepsilon^{\phi}$. Note that the left-hand side of \eqref{eq:pfconda} is an increasing function in $\varrho$, $\Delta^{\phi}$ and $\varepsilon^{\phi}$, when $\Delta^{\phi}$ and $\varepsilon^{\phi}$ are both small positive numbers. Thus, the maximum $\varepsilon^{\phi}$ (or equivalently, maximum $\delta$) to satisfy \eqref{eq:pfconda} can be obtained by taking the limit $\varrho \rightarrow 1$ and $\Delta^{\phi} \rightarrow 0$. It follows that
\begin{align*}
	\frac{2 \sin\varepsilon^{\phi}}{\cos^2(2\varepsilon^{\phi})} 
	\leq \frac{\zeta_{\min}}{\zeta_{\max}}
\end{align*}
and the inequality holds when $0<\delta=\sin\varepsilon^{\phi}\leq\delta_{\max}$, where $\delta_{\max}$ is the smallest positive root of the equation
\begin{align*}
	4 \delta^4 - 4 \delta^2 - 2 \frac{\zeta_{\max}}{\zeta_{\min}} \delta + 1 = 0\,.
\end{align*}
We give some numerical examples: $\delta_{\max}=0.318$ when ${\zeta_{\max}}/{\zeta_{\min}}=1$; $\delta_{\max}=0.096$ when ${\zeta_{\max}}/{\zeta_{\min}}=5$.


For (b): Consider a small angle $\sqrt{2a \varepsilon^{\phi}}$ as $\varepsilon^{\phi} \rightarrow 0$, where 
$a = {(2^{\beta+1}R^{2\beta}\zeta_{\max})}/{(r_0^{2\beta}\zeta_{\min})}$. 
The probability that all $\Nanc$ anchors locate in such a small angle of the $R \times R$ region is at most $(\sqrt{2a \varepsilon^{\phi}})^{\Nanc}$,
which goes to 0 at the rate of polynomial power $\Nanc/2$ as $\varepsilon^{\phi} \rightarrow 0$. Hence, we can find two anchors, $i$ and $i'$, whose angle separation is larger than $\sqrt{2a \varepsilon^{\phi}}$ and smaller than $\pi-\sqrt{2a \varepsilon^{\phi}}$ w.h.p.

We allocate the power equally on these two anchors, and it follows
\begin{align*}
	&\hspace{-1em} \tr \bigg\{ \Big( \sum_{j\in\Sanc} \xiwstkj \, \Pkj^* \, \Jr(\phiest_{kj})\Big)^{-1} \bigg\} \\
	\leq \,& \max_{\{\phi_{kj}\}}~\tr \bigg\{ \Big( \sum_{j\in\Sanc} \xiwstkj \, \breve{P}_{kj} \, \Jr(\phi_{kj})\Big)^{-1} \bigg\} \\
	\leq \,& \max_{\{\phi_{kj}\}}~\tr \bigg\{ \Big( \frac{\zeta_{\min}}{ (\sqrt{2} R)^{2\beta}} \frac{\Ptot}{2} \left(\Jr(\phi_{ki}) + \Jr(\phi_{ki'}) \right) \Big)^{-1} \bigg\} \\
	= \,&\, \frac{2^{\beta}R^{2\beta}}{\zeta_{\min}} \frac{2}{\Ptot} \frac{2}{\sin^2(\sqrt{2a \varepsilon^{\phi}}-2\varepsilon^{\phi})} \,.
\end{align*}

Finally, we need to show that
\begin{align*}
	\frac{2^{\beta}R^{2\beta}}{\zeta_{\min}} \frac{2}{\Ptot} \frac{2}{\sin^2(\sqrt{2a \varepsilon^{\phi}}-2\varepsilon^{\phi})}	
	\leq \frac{r_0^{2\beta}}{\zeta_{\max}} \frac{2}{\Ptot \, \sin\varepsilon^{\phi}}
\end{align*}
or equivalently,
\begin{align*}
	a \leq \frac{\sin^2(\sqrt{2a \varepsilon^{\phi}}-2\varepsilon^{\phi})} {\sin\varepsilon^{\phi}} \,.
\end{align*}
The above inequality holds as $\varepsilon^{\phi} \rightarrow 0$, since the limit of its right-hand side is $2a$.

Now, we extend the above proof to the multiple-agent case. In Section \ref{sec:distributed}, we decomposed the one-stage problem $\PRISDP$ into two-stage optimizations. Let $\rhokj^*$ and $\Pk^*$ denote the optimal solution of $\PRIDI$ and $\PRIDII$, respectively. Since the Stage I problem $\PRIDI$ is formulated for each single agent, we can show by the above proof that
\begin{align*}
	\sum_{j\in\Sanc} \xiwstkj \, \rhokj^* \, \Q(\phiest_{kj}) \succeq 0
\end{align*}
holds w.h.p. for agent $k$. Moreover, the optimal power allocation is given in \eqref{eq:optsolu} as $\Pkj^*=\rhokj^*\Pk^*$, where $\Pk^*$ obtained in Stage II does not affect $\rhokj^*$. Hence, we can show that the condition (\ref{eq:PRI-con-anc}) holds w.h.p. for multiple-agent networks.


\section{Proof of Proposition \ref{thm:Pk-close}}\label{apd:Pk-close}

The Lagrangian function is given by
\begin{align*}
	\mathcal{L}(\Pk,u_k,v) = \sum_{k\in\Sagt} \frac{T_k}{\Pk} - \sum_k u_k \Pk + v\bigg(\sum_{k\in\Sagt} \Pk - \Ptot\bigg)
\end{align*}
where $u_k,~v \geq 0$. The KKT conditions \cite{BerNedOzd:B03} can be derived as
\begin{align}
		\frac{\partial\mathcal{L}}{\partial{\Pk}} 
	= - \frac{T_k}{\Pk^2} - u_k + v &= 0
	\label{eq:KKT-1} \\
	u_k \Pk &= 0
	\notag \\
	v\bigg(\sum_{k\in\Sagt} \Pk - \Ptot\bigg) &= 0.
	\notag
\end{align}
Since $\Pk$ is always positive, we have $u_k=0$, which leads to $\Pk=\sqrt{{T_k}/{v}}$ in \eqref{eq:KKT-1}. Moreover, the objective is monotonically decreasing in $\Pk$, which implies the optimal allocation must use all the power resource, i.e., $\sum_{k\in\Sagt} \Pk = \Ptot$. Hence, the optimal solution is given by \eqref{eq:Pk-close}.

\section*{Acknowledgments}

The authors gratefully acknowledge Z.-Q.~Luo for his insightful discussion of the content of the paper, and H.~Yu and W.~Dai for their helpful suggestions and careful reading of the manuscript.

\bibliographystyle{IEEEtran}
\bibliography
{IEEEabrv,StringDefinitions,BiblioCV,WGroup,reference}

\vspace{-1em}
\begin{IEEEbiography}
	[{\includegraphics[width=1in,height=1.25in,clip,keepaspectratio]{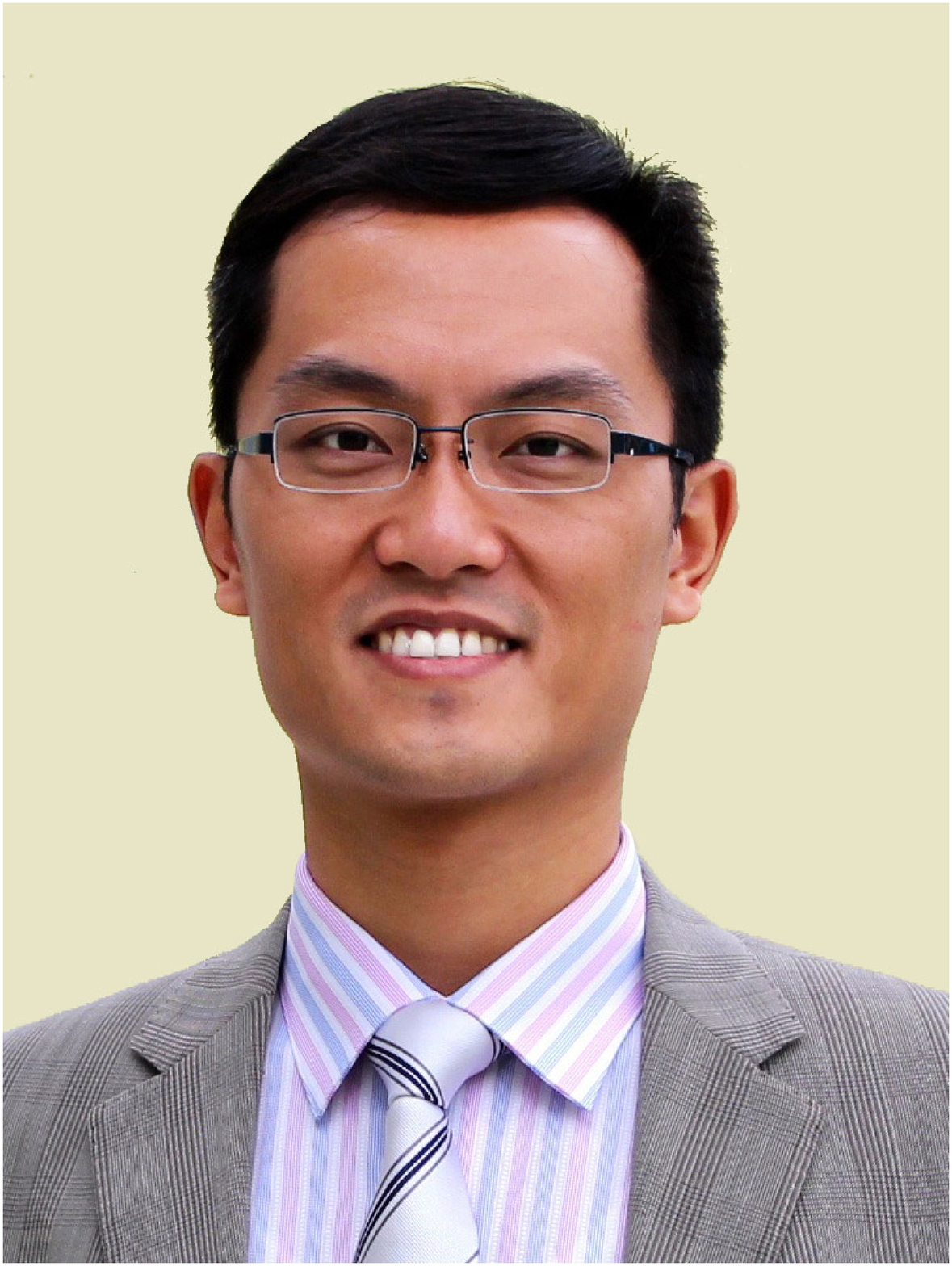}}]
	{William Wei-Liang Li} (S'09-M'12) received his Ph.D. degree in Information Engineering from the Chinese University of Hong Kong (CUHK), Hong Kong in 2012. He received a B.S. degree (with highest honor) in Automatic Control Engineering from Shanghai Jiao Tong University (SJTU), China in 2006.
	
	Since 2012, he has been with Department of Electrical and Computer Engineering, University of California, Santa Barbara, where he is currently a Postdoctoral Scholar. From 2006 to 2007, he was with the Circuit and System Laboratory, Peking University (PKU), China. From 2009 to 2011, he was a Visiting Student at the Wireless Communications and Network Science Laboratory, Massachusetts Institute of Technology (MIT). His main research interests include optimization and estimation theory, and their applications in wireless communications and networking.

	Dr. Li served as a Theory Session Chair of ACM MobiHoc S$^3$ Workshop in 2010, a Steering Committee Member of MIT 15th Annual LIDS Student Conference in 2010, a member of the Technical Program Committee (TPC) for the IEEE ICCVE in 2012, and the IEEE WCNC in 2013. 
During the four years of undergraduate study, he was consistently awarded the first-class scholarship, and graduated with highest honors from SJTU. He received the First Prize Award of the National Electrical and Mathematical Modelling Contest in 2005, the Global Scholarship for Research Excellence from CUHK in 2009, and the Award of CUHK Postgraduate Student Grants for Overseas Academic Activities and in 2009 and 2011.
\end{IEEEbiography}
 \vspace{-1em}
\begin{IEEEbiography}
	[{\includegraphics[width=1in, height=1.25in, clip, keepaspectratio]{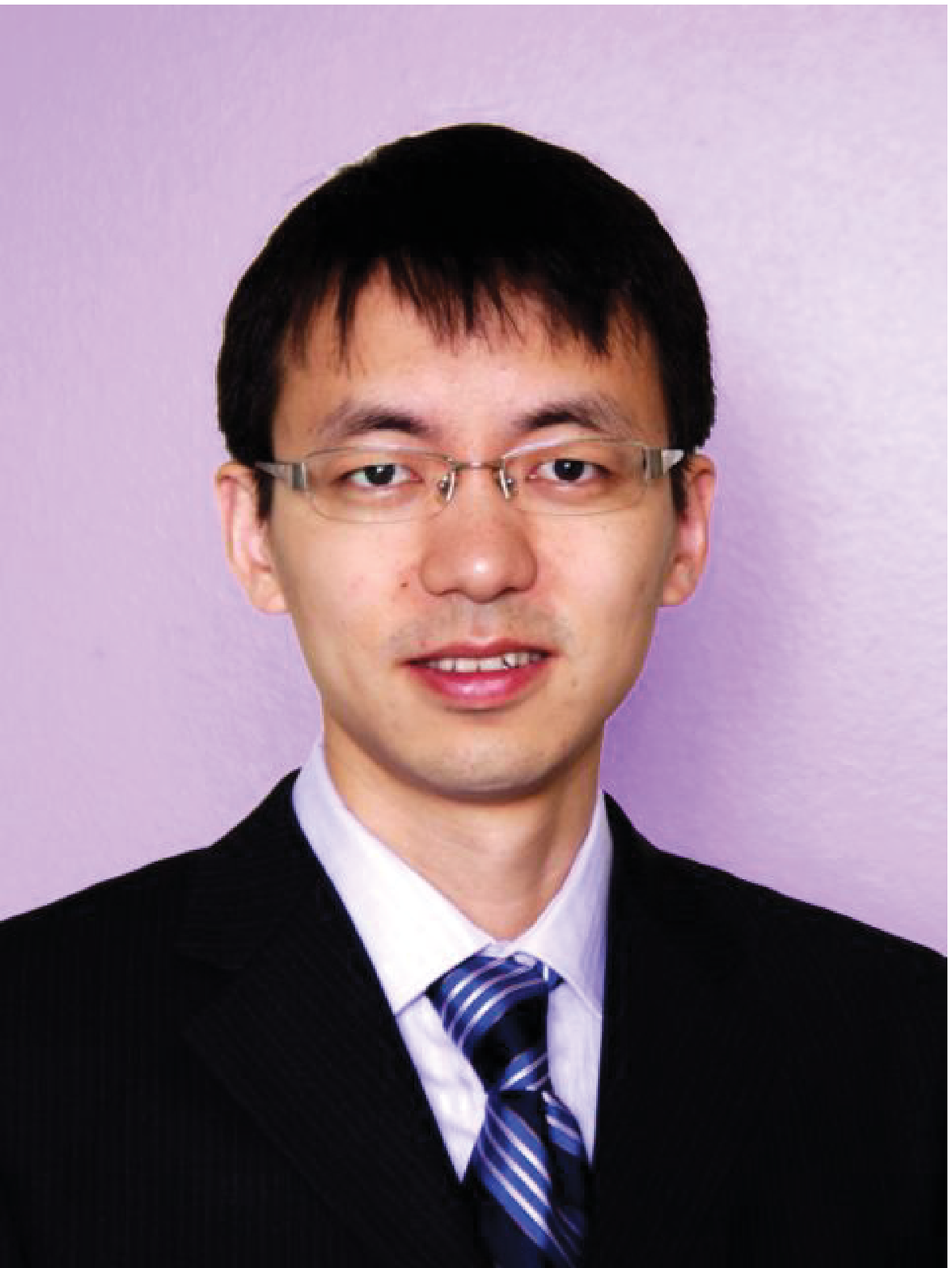}}]
	{Yuan Shen} (S'05) received his B.S.\ degree (with highest honor) from Tsinghua University, China, in 2005, and S.M.\ degree from the Massachusetts Institute of Technology (MIT), Cambridge, MA, in 2008, both in electrical engineering. 

	Since 2005, he has been with Wireless Communications and Network Science Laboratory at MIT, where he is now a Ph.D.\ candidate. He was with the Wireless Communications Laboratory at The Chinese University of Hong Kong in summer 2010, the Hewlett-Packard Labs in winter 2009, the Corporate R\&D of Qualcomm Inc.~in summer 2008, and the Intelligent Sensing Laboratory at Tsinghua University from 2003 to 2005. His research interests include statistical inference, network science, communication theory, and information theory. His current research focuses on network localization and navigation, generalized filtering techniques, resource allocation, intrinsic wireless secrecy, and cooperative networks.

	Mr.~Shen served as a member of the Technical Program Committee (TPC) for the IEEE Globecom in 2010--2013, the IEEE ICC in 2010--2013, the IEEE WCNC in 2009--2013, and the IEEE ICUWB in 2011--2013, and the IEEE ICCC in 2012. He is a recipient of the Marconi Society Paul Baran Young Scholar Award (2010), the MIT EECS Ernst A.~Guillemin Best S.M.~Thesis Award (first place) (2008), the Qualcomm Roberto Padovani Scholarship (2008), and the MIT Walter A. Rosenblith Presidential Fellowship (2005). His papers received the IEEE Communications Society Fred W. Ellersick Prize (2012) and three Best Paper Awards from the IEEE Globecom (2011), the IEEE ICUWB (2011), and the IEEE WCNC (2007).
\end{IEEEbiography}
\begin{IEEEbiography}
	[{\includegraphics[width=1in,height=1.25in,clip,keepaspectratio]{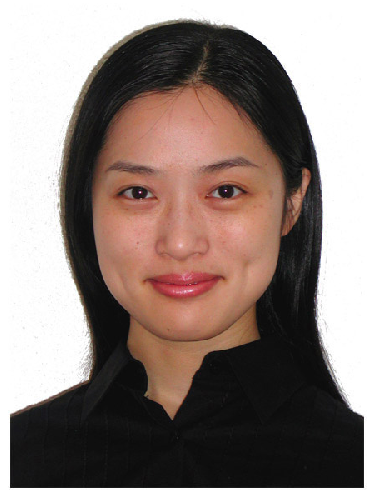}}]
	{Ying Jun (Angela) Zhang} (S'00-M'05-SM'11) received her PhD degree in Electrical and Electronic Engineering from the Hong Kong University of Science and Technology, Hong Kong in 2004. She received a B.Eng in Electronic Engineering from Fudan University, Shanghai, China in 2000.
	
	Since 2005, she has been with Department of Information Engineering, The Chinese University of Hong Kong, where she is currently an Associate Professor. She was with Wireless Communications and Network Science Laboratory at Massachusetts Institute of Technology (MIT) during the summers of 2007 and 2009. Her current research topics include resource allocation, convex and non-convex optimization for wireless systems, stochastic optimization, cognitive networks, MIMO systems, etc.
	
	Prof. Zhang is on the Editorial Boards of {\scshape IEEE Transactions on Wireless Communications}, {\scshape IEEE Transactions on Communications}, and Wiley Security and Communications Networks Journal. She was a Guest Editor of a Feature Topic in {\scshape IEEE Communications Magazine}. She has served as a TPC Vice-Chair of Wireless Communications Track of IEEE CCNC 2013, TPC Co-Chair of Wireless Communications Symposium of IEEE GLOBECOM 2012 Publication Chair of IEEE TTM 2011, TPC Co-Chair of Communication Theory Symposium of IEEE ICC 2009, Track Chair of ICCCN 2007, and Publicity Chair of IEEE MASS 2007. She is now a Co-Chair of IEEE ComSoc Multimedia Communications Technical Committee. She was an IEEE Technical Activity Board GOLD Representative, 2008 IEEE GOLD Technical Conference Program Leader, IEEE Communication Society GOLD Coordinator, and a Member of IEEE Communication Society Member Relations Council (MRC). She is a co-recipient of 2011 IEEE Marconi Prize Paper Award on Wireless Communications, the Annual Best Paper Award of IEEE Transactions on Wireless Communications. As the only winner from Engineering Science, she has won the Hong Kong Young Scientist Award 2006, conferred by the Hong Kong Institution of Science.
\end{IEEEbiography}
\begin{IEEEbiography}
	[{\includegraphics[width=1in,height=1.25in,clip,keepaspectratio]{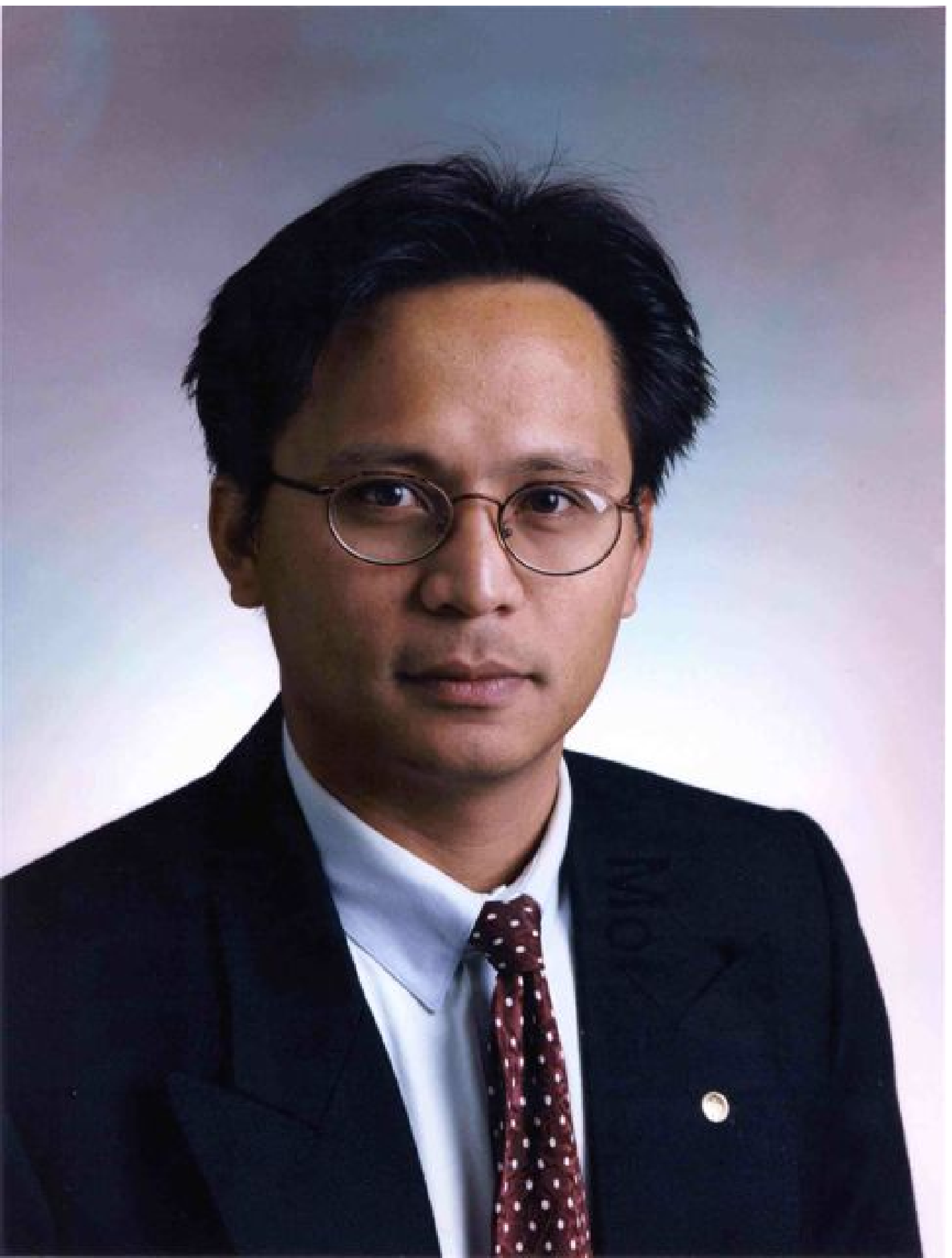}}]
	{Moe Z.~Win}
%
%
(S'85-M'87-SM'97-F'04) 
received 
	both the Ph.D.\ in 
		\switchAtwo{Electrical Engineering}
				{Electrical Engineering (under the supervision of Professor Robert A.\ Scholtz)}
	and M.S.\ in 
		\switchAtwo{Applied Mathematics}
				{Applied Mathematics (under the supervision of Professor Solomon W.\ Golomb)}
		as a Presidential Fellow at the University of Southern California (USC) in 1998.
He received an M.S.\ in Electrical Engineering from USC in 1989, 
	and a B.S.\ ({\em magna cum laude}) in Electrical Engineering from Texas A\&M University in 1987.

%
%
He is a Professor at the Massachusetts Institute of Technology (MIT).
Prior to joining MIT, he was at AT\&T Research Laboratories for five years 
and at the Jet Propulsion Laboratory for seven years.
%
%
His research encompasses fundamental theories, algorithm design, and
experimentation for a broad range of real-world problems.
His current research topics include
	network localization and navigation, 
	network interference exploitation, 
	intrinsic wireless secrecy,
	adaptive diversity techniques,
		\switchBtwo{and ultra-wide bandwidth systems.}
				{ultra-wide bandwidth systems,
				optical transmission systems, 
				and 
				space communications systems.}
		
%
%
Professor Win is 
	an elected Fellow of the AAAS, the IEEE, and the IET, 
	and was an IEEE Distinguished Lecturer.
He was honored with two IEEE Technical Field Awards: 
	the IEEE Kiyo Tomiyasu Award (2011) 
		\switchBtwo{and}
				{for ``fundamental contributions to high-speed reliable communications over optical and wireless channels''  and}
	the IEEE Eric E. Sumner Award 
		\switchBtwo{(2006, jointly with R.\ A.\ Scholtz).}
				{(2006, jointly with R.\ A.\ Scholtz) for ``fundamental contributions to high-speed reliable communications over optical and wireless channels.''}
%
%
Together with students and colleagues, his papers have received numerous awards including
	\switchBtwo{the IEEE Communications Society's Stephen O.\ Rice Prize (2012),
			the IEEE Aerospace and Electronic Systems Society's M.\ Barry Carlton Award (2011),
			the IEEE Communications Society's Guglielmo Marconi Prize Paper Award (2008),
    			and the IEEE Antennas and Propagation Society's Sergei A.\ Schelkunoff Transactions Prize Paper Award (2003).}
			{the IEEE Communications Society's Stephen O.\ Rice Prize (2012),
			the IEEE Communications Society's William R.\ Bennett Prize (2012),
			the IEEE Communications Society's Fred W.\ Ellersick Prize (2012),
			the IEEE Communications Society's Leonard G.\ Abraham Prize (2011),
			the IEEE Aerospace and Electronic Systems Society's M.\ Barry Carlton Award (2011),
			the IEEE Communications Society's Guglielmo Marconi Prize Paper Award (2008),
    			and the IEEE Antennas and Propagation Society's Sergei A.\ Schelkunoff Transactions Prize Paper Award (2003).}
%
%
Highlights of his international scholarly initiatives are
	the Copernicus Fellowship (2011),
	the Royal Academy of Engineering Distinguished Visiting Fellowship (2009),
	and
	the Fulbright
		\switchBtwo{Fellowship (2004).}
				{Foundation Senior Scholar Lecturing and Research Fellowship (2004).}
%
%
Other recognitions include
	\switchBtwo{the Laurea Honoris Causa from the University of Ferrara (2008),	
				the Technical Recognition Award of the IEEE ComSoc Radio Communications Committee (2008),
        				and the U.S. Presidential Early Career Award for Scientists and Engineers (2004).}	
			{the Outstanding Service Award of the IEEE ComSoc Radio Communications Committee (2010),
				the Laurea Honoris Causa from the University of Ferrara, Italy (2008),
				the Technical Recognition Award of the IEEE ComSoc Radio Communications Committee (2008),
        				the Wireless Educator of the Year Award (2007),
				the U.S. Presidential Early Career Award for Scientists and Engineers (2004), 
				the AIAA Young Aerospace Engineer of the Year (2004),
				and
				the Office of Naval Research Young Investigator Award (2003).}

%
%
Dr.\ Win is an elected Member-at-Large on the IEEE Communications Society Board of Governors (2011--2013).
He was
    the chair (2004--2006) and secretary (2002--2004) for
        the Radio Communications Committee of the IEEE Communications Society.
%
%
\switchBtwo{Over the last decade, he has organized and chaired numerous international conferences.} 
		{He served as
			the Technical Program Chair for
				the IEEE Wireless Communications and Networking Conference (2009),
				the IEEE Conference on Ultra Wideband (2006),
				the IEEE Communication Theory Symposia of ICC (2004) and Globecom (2000),
				 and
				the IEEE Conference on Ultra Wideband Systems and Technologies (2002);
			Technical Program Vice-Chair for
				the IEEE International Conference on Communications (2002); and
			the Tutorial Chair for
				ICC (2009) and
				the IEEE Semiannual International Vehicular Technology Conference (Fall 2001).}
%
%
He is currently
	an Editor-at-Large for the 
	{\scshape IEEE Wireless Communications Letters},
	and serving on the Editorial Advisory Board for the 
	{\scshape IEEE Transactions on Wireless Communications}.
He served as Editor (2006--2012) for
	\switchBtwo{the {\scshape IEEE Transactions on Wireless Communications},
				and served as 
				Area Editor (2003--2006) and Editor (1998--2006)
				for the {\scshape IEEE Transactions on Communications}.}      	
			  {the {\scshape IEEE Transactions on Wireless Communications}.
				He also served as the
	    			Area Editor for {\em Modulation and Signal Design} (2003--2006),
	    			Editor for {\em Wideband Wireless and Diversity} (2003--2006), and
	    			Editor for {\em Equalization and Diversity} (1998--2003),
	        			all for the {\scshape IEEE Transactions on Communications}.}
He was Guest-Editor
        for the
        {\scshape Proceedings of the IEEE}
		\switchBtwo{(2009)}
				{(Special Issue on UWB Technology \& Emerging Applications -- 2009)} and
        {\scshape IEEE Journal on Selected Areas in Communications}
        		 \switchBtwo{(2002).}
		 		{(Special Issue on Ultra\thinspace-Wideband Radio in Multiaccess Wireless Communications  -- 2002).}

\end{IEEEbiography}

\end{document}